\documentclass[10pt]{article}

\usepackage[T1]{fontenc}
\usepackage{textcomp}
\usepackage{palatino}
\usepackage{mathpazo}
\usepackage{stmaryrd}

\usepackage{hyperref}
\hypersetup{pdfpagemode=UseNone}

\usepackage{amsfonts}
\usepackage{amssymb}
\usepackage{amsmath}
\usepackage{latexsym}
\usepackage{amsthm}
\usepackage{eepic}
\usepackage{sectsty}
\usepackage[usenames]{color}

\newtheorem{theorem}{Theorem}
\newtheorem{lemma}{Lemma}

\newtheorem{cor}{Corollary}
\newtheorem{definition}{Definition}

\newtheorem{fact}{Fact}

\newtheorem{framework}{Framework}

\newcommand{\tinyspace}{\mspace{1mu}}
\newcommand{\microspace}{\mspace{0.5mu}}

\newcommand{\norm}[1]{\left\lVert\tinyspace#1\tinyspace\right\rVert}
\newcommand{\snorm}[1]{\lVert\tinyspace#1\tinyspace\rVert}

\newcommand{\ceil}[1]{\left\lceil #1 \right\rceil}

\newcommand{\tr}{\operatorname{Tr}}

\renewcommand{\vec}{\operatorname{vec}}

\newcommand{\ip}[2]{\left\langle #1 , #2\right\rangle}

\def\({\left(}
\def\){\right)}
\def\I{\mathbb{1}}

\newcommand{\fid}{\mathcal{F}}
\newcommand{\setft}[1]{\mathrm{#1}}
\newcommand{\lin}[1]{\setft{L}\left(#1\right)}
\newcommand{\density}[1]{\setft{D}\left(#1\right)}
\newcommand{\unitary}[1]{\setft{U}\left(#1\right)}

\newcommand{\herm}[1]{\setft{Herm}\left(#1\right)}
\newcommand{\pos}[1]{\setft{Pos}\left(#1\right)}

\def \lket {\left|}
\def \rket {\right\rangle}
\def \lbra {\left\langle}
\def \rbra {\right|}
\newcommand{\ket}[1]{\lket\microspace #1 \microspace\rket}
\newcommand{\bra}[1]{\lbra\microspace #1 \microspace\rbra}

\newenvironment{mylist}[1]{\begin{list}{}{
    \setlength{\leftmargin}{#1}
    \setlength{\rightmargin}{0mm}
    \setlength{\labelsep}{2mm}
    \setlength{\labelwidth}{8mm}
    \setlength{\itemsep}{0mm}}}
    {\end{list}}

\newcommand{\class}[1]{\textup{#1}}

\newcommand{\equil}[1]{#1^\star}
\newcommand{\ceq}[1]{Equation [#1]}

\def\X{\mathcal{X}}
\def\Y{\mathcal{Y}}
\def\Z{\mathcal{Z}}

\def\A{\mathcal{A}}
\def\B{\mathcal{B}}
\def\V{\mathcal{V}}

\def\C{\mathcal{C}}

\def\M{\mathcal{M}}
\def\N{\mathcal{N}}

\begin{document}

\title{\bf Parallelized Solution to Semidefinite Programmings in Quantum Complexity Theory}

\author{%
  Xiaodi Wu \footnote{The work was completed when the author was visiting the Institute for Quantum Computing , University of Waterloo.}   \\
  \it \small Institute for Quantum Computing,  University of
  Waterloo,  Ontario, Canada, \\
  \it \small Department of Electrical Engineering and Computer Science, University of Michigan, Ann Arbor,
  USA \\
}

\date{July, 2010}

\maketitle

\begin{abstract}
In this paper we present an equilibrium value based framework for
solving SDPs via the multiplicative weight update method which is
different from the one in Kale's thesis~\cite{Kale07}. One of the
main advantages of the new framework is that we can guarantee the
convertibility from approximate to exact feasibility in a much more
general class of SDPs than previous result. Another advantage is the
design of the oracle which is necessary for applying the
multiplicative weight update method is much simplified in general
cases. This leads to an alternative and easier solutions to the SDPs
used in the previous results
\class{QIP(2)}$\subseteq$\class{PSPACE}~\cite{JainUW09} and
\class{QMAM}=\class{PSPACE}~\cite{JainJUW09}. Furthermore, we
provide a generic form of SDPs which can be solved in the similar
way. By parallelizing every step in our solution, we are able to
solve a class of SDPs in \class{NC}. Although our motivation is from
quantum computing, our result will also apply directly to any SDP
which satisfies our conditions.

In addition to the new framework for solving SDPs, we also provide a
novel framework which improves the range of equilibrium value
problems that can be solved via the multiplicative weight update
method. Before this work we are only able to calculate the
equilibrium value where one of the two convex sets needs to be the
set of density operators. Our work demonstrates that in the case
when one set is the set of density operators with further linear
constraints, we are still able to approximate the equilibrium value
to high precision via the multiplicative weight update method.

\end{abstract}

\newpage

\section{Introduction} \label{sec:introduction}
Semidefinite programming (SDP) is a relatively new field of
optimization which grew up in the
1990s~\cite{Alizadeh95,BoydV04,Lovasz03,VandenbergheB96,deKlerk02}.
Despite of the short time since its introduction, SDP has proved
useful in many different contexts. Especially, there are many
applications of SDPs in the field of theoretical computer science,
like the design of approximation algorithm~\cite{Vazirani01} and the
recent application to Unique Game
Conjecture~\cite{Steurer10,Raghavendra08} and simulating quantum
complexity
classes~\cite{Gutoski05,GutoskiW05,GutoskiW07,JainJUW09,JainUW09,KitaevW00,Watrous09b}.

SDPs are in fact a special case of conic programming and there
exists polynomial algorithms to solve any SDP instance (like the
interior point method~\cite{Alizadeh95, deKlerk02}). Thus, any
problem which can be modeled or approximated as a SDP is considered
to have efficient solutions. However, on the other side, our
understanding of SDPs is far less than our understanding of Linear
programming(LP), another typical optimization method used a lot in
practice. One thing is about the running time : although any SDP
instance can be solved in polynomial time theoretically, it is much
slower than LP's solution in practice. Moreover, the generic
algorithm for SDP are always used as a black box. Hence it is rarely
seen that we can employ the duality or the structure of SDPs while
the duality of LPs inspires lots of new algorithm designs.

A generic \emph{primal-dual} method for SDP problem (under certain
conditions) was introduced by Arora
\emph{et.al.}~\cite{AroraHK05a,AroraHK05b, AroraK07} to overcome the
difficulties mentioned above. The generic method exploits a generic
framework (or meta-algorithm) called the \emph{multiplicative weight
update method}. Similar frameworks were studied
~\cite{Fle00,FreundS99,GargK98,PlotkinST91,Khandekar04,WarmuthK06,Young95}
for different purposes before. This new generic method turns out to
be very useful and successful in several contexts. In the
paper~\cite{AroraK07} (see also Kale's Phd thesis~\cite{Kale07})
where it was originally proposed, this generic method improves the
upper bounds of running time of many approximation algorithms. In
addition to that, this generic method was introduced by Watrous
\emph{et. al.} to the field of quantum computation and successfully
proved
\class{QIP(2)}$\subseteq$\class{PSPACE}~\cite{JainUW09} and
\class{QIP}$=$\class{PSPACE}~\cite{JainJUW09}.

A \emph{semidefinite program} over $\X$ and $\Y$ (shown below) is
specified by a triple $(\Psi, A, B)$ where $\Psi : \lin{\X}
\rightarrow \lin{\Y}$ is a Hermiticity preserving super-operator and
$A\in \herm{\X}$ and $B \in \herm{\Y}$. This form which was obtained
in~\cite{Watrous08} is somehow different from but equivalent to the
standard form. Let $\alpha, \beta$ be the optimum values of the
primal and dual programs respectively. One important property of the
semidefinite program, called the \emph{duality}, implies that
$\alpha \leq \beta$ and the equality will hold in some situation.

\begin{center}
  \begin{minipage}{2in}
    \centerline{\underline{Primal problem}}\vspace{-7mm}
    \begin{align*}
      \text{maximize:}\quad & \ip{A}{X}\\
      \text{subject to:}\quad & \Psi(X) \leq B,\\
      & X\in\pos{\X}.
    \end{align*}
  \end{minipage}
  \hspace*{12mm}
  \begin{minipage}{2in}
    \centerline{\underline{Dual problem}}\vspace{-7mm}
    \begin{align*}
      \text{minimize:}\quad & \ip{B}{Y}\\
      \text{subject to:}\quad & \Psi^{\ast}(Y) \geq A,\\
      & Y\in\pos{\Y}.
    \end{align*}
  \end{minipage}
\end{center}

In practice we usually consider the SDPs whose optimum values are
within a small range. For those SDPs, instead of directly
calculating the optimum value, we usually consider the
\emph{feasibility problem} first. Once the feasibility problem is
solved, we could use the binary search in that range to find the
optimum value. For any instance of SDP and a guess value $c$ , the
feasibility problem is defined to be

\begin{center}
  \centerline{\underline{Feasibility Problem}}\vspace{-7mm}
    \begin{align*}
      \text{ask whether:}\quad & \ip{A}{X} \geq c\\
      \text{subject to:}\quad & \Psi(X) \leq B,\\
      & X\in\pos{\X}.
    \end{align*}
\end{center}

Intuitively, the primal-dual method in Kale's Thesis~\cite{Kale07}
for the feasibility problem contains the following three
ingredients. First, update via the multiplicative weight update
method. Second, the existence of an \emph{efficient width-bounded}
oracle $\mathcal{O}_1$. Finally, this method will generate an
approximately feasible solution $Y$ to the dual problem such that
$\ip{B}{Y}\leq c$ . By approximate feasibility, we mean
\begin{equation} \label{eqn:approx_exact_1}
 \forall \rho \in \density{\X}, \ip{\Psi^\ast(Y)-A}{\rho} \geq
 -\epsilon
\end{equation}
where the $\density{\X}$ denotes the set of density operators over
the space $\X$ and $\epsilon$ is some small constant. In order to
solve the feasibility problem, we need to get an exact dual feasible
solution $\tilde{Y}$ such that $\ip{B}{\tilde{Y}} \leq
(1+\varepsilon)c$ from this. We will refer this as the
\emph{convertibility from approximate to exact feasibility}. After
executing a combination of those three ingredients, this method will
either return a feasible solution $\tilde{X}$ to the primal problem
with object function value at least $c$ or a feasible solution
$\tilde{Y}$ such that $\ip{B}{\tilde{Y}} \leq (1+\varepsilon)c$. The
latter case will imply $\alpha\leq \beta\leq (1+\epsilon)c$ by the
duality of SDPs. A detailed description of this procedure can be
found in Appendix~\ref{sec:MMW_survey}.


Another important value which can be calculated via the
multiplicative weight update method is the \emph{equilibrium value}
of zero-sum games~\cite{vN28} and its generalizations (like,
\cite{Hazan06}). Particulary, we consider the value $\lambda$,
\[
  \lambda=  \min_{x \in X} \max_{y \in Y} f(x,y)= \max_{y \in Y} \min_{x \in X} f(x,y)
\]
for some \emph{convex-concave} function $f$ (see definition in
Appendix~\ref{sec:MMW_survey}) over $X \times Y$ where $X, Y$ are
convex compact sets. Again, this method involves the update via the
multiplicative weight update method and an \emph{efficient
width-bounded} but functionally different oracle $\mathcal{O}_2$ as
main ingredients. However, for equilibrium value, there is no
requirement for the convertibility from approximate to exact
feasibility. Under several other conditions as well (see details in
Theorem~\ref{thm:mmw_generic_equilibrium} and
Appendix~\ref{sec:MMW_survey}), such value $\lambda$ can be
approximated to high precision efficiently. It should be noted here
the multiplicative weight update method plays a quite different role
from the one in the solution to SDPs above. This difference has led
to an alternative and easier proof of
\class{QIP}=\class{PSPACE}~\cite{Wu10}. It was also applied to the
proof of \class{QRG(1)}$\subseteq$\class{PSPACE}~\cite{JainW09}.

One important advantage of the solutions to SDPs or equilibrium
value based on the multiplicative weight update method is that we
can easily implement the algorithm in parallel. Precisely, this is
because the fundamental operations of matrices and singular value
decomposition of matrices~\cite{vzGathen93} can be implemented with
high accuracy in \class{NC}. This trick was widely exploited in the
recent progress of quantum complexity
theory~\cite{Wu10,JainJUW09,JainUW09,JainW09}.

In this paper we will demonstrate how the equilibrium value can be
related to the feasibility problem (then, SDPs). It will then
suffice to make use of the solution to the equilibrium value to
solve the feasibility problem. Following this idea, we will provide
an alternative and easier solution to the SDPs used for
\class{QIP(2)}~\cite{JainUW09} and \class{QMAM}~\cite{JainJUW09}.
Moreover, we provide a generic form and conditions under which any
feasibility problem can be solved in the same way. By parallelizing
each step in the solution, we are able to show any feasibility
problem can be solved in \class{NC} as well. Precisely, we consider
the feasibility problem as follows.

\begin{center}
  \centerline{\underline{Feasibility Problem}}\vspace{-6mm}
    \begin{align*}
      \text{ask whether:}\quad & \ip{A}{X} \geq c\\
      \text{subject to:}\quad & \Psi(X) \leq B,\\
      & X\in\density{\X}.
    \end{align*}
\end{center}

This feasibility problem is very similar to the most general version
above. The only difference is we replace the condition $X \in
\pos{\X}$ by $X \in \density{\X}$. This change corresponds to the
trace bound $\tr X =R$ for some const $R$ which commonly appears in
applications of SDPs. Under certain conditions our constraint on $X
\in \density{\X}$ is equivalent to the trace bound $\tr X \leq R$.
This implies the feasibility problem in our consideration is very
general and could cover many instances in practical use. Although
our original motivation is from quantum computation, our result also
works well for any SDPs which can be converted to the form in our
consideration.

The concept of equilibrium values will be related if we imagine a
two-player game to solve the feasibility problem. Assume there is a
primal player who wants to provide you a feasible solution $X$ to
prove the original problem is feasible. On the contrary, the dual
player who wants to disprove the feasibility will try to find where
the constraints on $X$ are violated. This is different from Kale's
method where the disproof of the feasibility is by getting some
feasible solution to the dual problem and then making use of the
duality of SDPs. Precisely, we define the following convex-concave
function $f$ to capture the two-player game.

\begin{framework} \label{framework_1}
Let function $f$ be
\begin{equation} \label{eqn:framework1}
 f(X,\Pi)= \ip{\left(
                                                          \begin{array}{cc}
                                                           c- \ip{A}{X}   &  \\
                                                              &  \Psi(X)-B \\
                                                          \end{array}
                                                        \right)
 }{\Pi}
\end{equation}
over the set $\density{\X} \times T$ where $T=\{ \Pi: 0\leq \Pi \leq
\I_{\Y \oplus \mathbb{C}} \}$. Let the equilibrium value
$\equil{\lambda}$ be
\[
 \equil{\lambda}=\min_{X \in \density{X}}\max_{\Pi \in T} f(X, \Pi)=\max_{\Pi \in T}\min_{X \in \density{X}} f(X, \Pi)
\]
\end{framework}

The relation between the equilibrium value $\equil{\lambda}$ and the
feasibility of the original problem is captured by the following
theorem.

\begin{theorem} \label{thm:equilibrum_feasibility}
The original problem is feasible if and only if $\equil{\lambda}
\leq 0$.
\end{theorem}

This formulation is inspired by similar formulations~\cite{Hazan06}
used for convex optimization and
Theorem~\ref{thm:equilibrum_feasibility} follows easily from the
argument in~\cite{Hazan06}. Given the fact that equilibrium value
$\equil{\lambda}$ can only be calculated approximately, we still
need to do the conversion from approximately to exactly feasible
solutions. Due to this reason, we make an important change to the
old formulation in~\cite{Hazan06}. Namely, we choose $T$ to be $\{
\Pi: 0\leq \Pi \leq \I_{\Y \oplus \mathbb{C}} \}$ rather than the
set of density operators over the space $\Y \oplus \mathbb{C}$.
Assume we approximate $\equil{\lambda}$ to precision $\epsilon$ and
the return value implies $\equil{\lambda}$ is in an interval
containing 0. Or equivalently, we encounter the situation
\begin{equation} \label{eqn:approx_exact_2}
\forall \quad \Pi \in T, \ip{\left(
                                                          \begin{array}{cc}
                                                           c-\ip{A}{X}   &  \\
                                                              &  \Psi(X)-B \\
                                                          \end{array}
                                                        \right)
 }{\Pi}\leq \epsilon
\end{equation}
and we want to convert $X$ into some exact feasible solution without
changing the object function a lot. This is very similar to the
situation captured by \ceq{\ref{eqn:approx_exact_1}}. However, the
difficulty in making this conversion happen is different. By
applying the approximate feasible result in
\ceq{\ref{eqn:approx_exact_1}}, we can only expect to solve the
problems~\cite{JainUW09,JainJUW09} when $A=\I_\X$ efficiently in
general. This is because the density operator $\rho$ in
\ceq{\ref{eqn:approx_exact_1}} implies the approximate feasibility
only holds in term of $\mathcal{L}_\infty$ norm. On the other side,
in order to remain the object function which is a inner product
almost unchanged, the approximate feasibility needs to hold in a
stronger sense like $\mathcal{L}_1$ norm. Otherwise only trivial
cases like $A=\I_\X$ can be solved.

Our new approximate feasibility result in
\ceq{\ref{eqn:approx_exact_2}} overcomes such difficulty and
provides a method to do the conversion from the approximate to exact
feasibility in much more general cases. Especially we will
demonstrate when the super-operator $\Psi$ is \emph{partial trace}
or its generalizations, our approximate feasibility result works
very well to make the conversion happen. This kind of constraints is
very powerful because it is the only type of constraints we need in
many SDPs in quantum computation. In addition to that, we still have
the freedom to choose $T$ in order to  meet the requirement of new
types of constraints. Such freedom is a huge advantage over the
primal-dual method since \ceq{\ref{eqn:approx_exact_1}} is the only
result can be expected from the primal-dual method.

Another advantage of the Framework~\ref{framework_1} is that the
oracle $\mathcal{O}_2$ which will be required to compute the
equilibrium value $\equil{\lambda}$ can be easily designed.
Precisely, since $T=\{ \Pi: 0\leq \Pi \leq \I_{\Y \oplus \mathbb{C}}
\}$, we can simply get the other part's spectral decomposition and
let $T$ be the projection onto the positive eigenspace of it.
Finally, under the conditions of
Theorem~\ref{thm:mmw_generic_equilibrium}, we can implement the
whole algorithm in \class{NC}.

Besides the application to the feasibility problem, the idea of
Framework~\ref{framework_1} can also improve our ability on
calculating equilibrium over some non-density operator set. Consider
some convex-concave function $h$ over $X \times Y$ where $X$ is the
set of density operators but with some constraint $\Phi$. Precisely,

\begin{framework} \label{framework_2}
If we consider the equilibrium value $\mu$,
\[
  \mu=\min_{x \in X} \max_{y \in Y} h(x,y)=\max_{y \in Y} \min_{x
  \in X } h(x,y)
\]
where $X =\{\rho : \rho \in \density{\X}, \Phi(\rho) \leq B\}$ and
$Y$ is any other compact convex set. It will be useful to consider
the equilibrium value $\mu'$ as follows
\[
  \mu'=\min_{x \in X'} \max_{\{y,\Pi\} \in Y'} h(x,y)+\alpha
  \ip{\Phi(\rho)-B}{\Pi}= \max_{\{y,\Pi\} \in Y'} \min_{x \in X'} h(x,y)+\alpha
  \ip{\Phi(\rho)-B}{\Pi}
\]
where $X=\density{\X}$ , $Y'=Y \times T$ ($T=\{ \Pi: 0\leq \Pi \leq
\I\}$) and $\alpha$ is any factor.
\end{framework}

It is easy to see that the idea of Framework~\ref{framework_2} or
especially the term $\ip{\Phi(\rho)-B}{\Pi}$ is to penalize when
$\rho$ does not satisfy the constraint $\Phi(\rho)\leq B$.
Furthermore, the penalization is weighted according to the factor
$\alpha$. By using the game value of one restricted model of
one-round quantum refereed game as an example, we can demonstrate
when there are two promise values of $\mu$ with large gap, such gap
can be transferred to the new value $\mu'$. Thus, it will be
sufficient to calculate the value of $\mu'$ in order to distinguish
between two promises of the original value $\mu$. This improves the
range of problems which can be solved by the multiplicative weight
update method since so far we can only calculate the equilibrium
value when $X$ is the set of density operators (up to a factor).
This will also give a binary search method for calculating the
equilibrium value to some precision by artificially assuming two
promise values.

The rest of the paper is organized as follows. Most of the
preliminaries can be found in Appendix~\ref{sec:appendix_prelim}. We
also leave the lemmas and theorems which will be directly used in
Section~\ref{sec:prelim}. The two examples for
Framework~\ref{framework_1} will be demonstrated in
Section~\ref{sec:qip2} (\class{QIP(2)}) and in
Appendix~\ref{sec:qmam} (\class{QMAM}) respectively. One restricted
version of one-round quantum refereed game will be discussed in
Section~\ref{sec:qrg2} to demonstrate the power of
Framework~\ref{framework_2}. We will conclude the paper with further
discussions and open problems in Section~\ref{sec:conclusion}.

There are two points to make clear before the readers move on to the
next section. First, when we consider the feasibility problem or
equilibrium value directly, we always refer the size of the SDP or
the function with equilibrium value as the input size. However, when
we consider the quantum complexity classes, the SDP or the function
with equilibrium value will have exponential size in term of its
actual input size $|x|$. Second, we will not take care of the
precision issues with the \class{NC} implementation in the main part
of this paper. Instead, we will assume such implementation can be
made exactly and deal with them in
Appendix~\ref{sec:precision_issue}.

\section{Preliminaries} \label{sec:prelim}
In order to make this paper self-contained, we try to provide brief
surveys on each topic related to our paper. However, most of them
will be put in the appendix due to limited space. Precisely, we will
introduce the fundamentals of quantum information in
Appendix~\ref{sec:prelim_fundamentals_qi}. Useful facts on
\class{NC} and parallel matrix operations are stated in
Appendix~\ref{sec:NC_parellel}. The multiplicative weight update
method and its application to calculating the equilibrium value and
SDPs (the primal-dual method) are surveyed in
Appendix~\ref{sec:MMW_survey}.

\subsection{Useful Lemma and Facts in Quantum Information}
\label{sec:lemma_facts}

 The following are some important lemmas
about purification and fidelity which are useful in our proof later.
The proofs of following results are put in
Appendix~\ref{sec:prelim_fundamentals_qi}.

\begin{lemma} \label{lemma:purification_fidelity}
Given any two density operators $\rho_1,\rho_2$ over the space $\A$,
and another density operator $\sigma_1$ over the space $\A \otimes
\B$ such that $\tr_B \sigma_1= \rho_1$, then there exists another
density operator $\sigma_2$ over the space $\A \otimes  \B$ for
which that $\tr_B \sigma_2 =\rho_2$ and
$\mathcal{F}(\rho_1,\rho_2)=\mathcal{F}(\sigma_1,\sigma_2)$.
\end{lemma}

\begin{lemma} \label{lemma:compute_purification}
In addition to the result in
Lemma~\ref{lemma:purification_fidelity}, we can compute the
classical representation of $\sigma_2$ as required above given the
classical representations of $\rho_1,\rho_2$ and $\sigma_1$ in
\class{NC} where the input size refers to the size of the matrices.
\end{lemma}

\begin{lemma} \label{lemma:trace_bound}
Given two density operators $\rho_1, \rho_2 \in \density{\A}$, and
their purifications $\sigma_1, \sigma_2 \in \density{\A \otimes \B}$
in space $\A \otimes \B$ respectively, for which
$\fid(\rho_1,\rho_2)=\fid(\sigma_1, \sigma_2)$, let
\[
  s= \frac{1}{2}\snorm{\rho_1-\rho_2}_1 \text{ and }
  t=\frac{1}{2}\snorm{\sigma_1-\sigma_2}_1
\]
Then we have the following inequalities
\[
  1-s\leq \sqrt{1-t^2} \text {  and } 1-t \leq \sqrt{1-s^2}
\]
\end{lemma}

\begin{lemma} \label{lemma:admissable_connect}
Given two density operators $\rho_1,\rho_2 \in \density{\A \otimes
\B}$ where $\rho_1$ represents a pure state, there exists an
admissable quantum channel $\Phi : \lin{\A} \rightarrow \lin{\A}$
such that $\Phi \otimes \I_{\lin{\B}}(\rho_1)=\rho_2$ if and only if
$\tr_\A (\rho_1) =\tr_\A (\rho_2)$.
\end{lemma}

\begin{lemma} \label{lemma:trace_norm}
Given any Hermitian operator $A$ such that $\tr A=0$, then we have
\[
  \max_{\Pi: 0\leq \Pi \leq \I} \ip{A}{\Pi}=\frac{1}{2}\snorm{A}_1 \text{  and }  \min_{\Pi: 0\leq \Pi \leq \I} \ip{A}{\Pi}=-\frac{1}{2}\snorm{A}_1
\]
\end{lemma}

\subsection{Multiplicative Weights Update Method}
\label{sec:MMW_generic_equilibrium}

The detailed discussion of the multiplicative weight update method
is provided in Appendix~\ref{sec:MMW_survey}. However, we will
demonstrate the particular algorithm to calculate the equilibrium
value in the form of \ceq{\ref{eqn:framework1}} in this section.
Precisely, we consider the

\begin{equation} \label{eqn:negative_equilibrium}
 \equil{\lambda}=\min_{\rho \in \density{\X}} \max_{\Pi \in T} \ip{\left(
                                                          \begin{array}{cc}
                                                           c- \ip{A}{\rho}   &  \\
                                                              & \Psi(\rho)-B \\
                                                          \end{array}
                                                        \right)
 }{\Pi}
\end{equation}

The existence of the equilibrium value is implied before. Thus we
only need to see when a \class{NC} algorithm will exist to calculate
the $\lambda$ approximately to high precision. To ease the
description of the algorithm, let $S(\rho)$ defined to be
\[
  S(\rho)=\left(
                                                          \begin{array}{cc}
                                                           c- \ip{A}{\rho}   &  \\
                                                              & \Psi(\rho)-B \\
                                                          \end{array}
                                                        \right)
\]
for any $\rho \in \density{\X}$. Choose $N(\Pi)$ to be the raw loss
matrix such that for any $\rho \in \density{\X}$ and $\Pi=\left(
                                                            \begin{array}{cc}
                                                              p &  \\
                                                                & P  \\
                                                            \end{array}
                                                          \right)
 \in T$,
we have
\begin{equation} \label{eqn:S_def}
  \ip{S(\rho)}{\Pi}=\ip{\rho}{N(\Pi)}
\end{equation}
It is easy to see that we can choose $N(\Pi)$ to be
\begin{equation} \label{eqn:N_def}
  N(\Pi)=-pA +\Psi^\ast(P) + (pc- \ip{B}{P})\I_\X
\end{equation}

\begin{figure}[t]
\noindent\hrulefill
\begin{mylist}{8mm}
\item[1.]
Let $\varepsilon=\frac{\delta}{4r}$ and $T=\ceil{\frac{16r^2 \ln
D}{\delta^2}}$. Also let $W^{(1)}=\I_{\X}$, $D=\dim{(\X)}$.

\item[2.]
Repeat for each $t = 1,\ldots,T$:

\begin{mylist}{8mm}
\item[(a)]
Let $\rho^{(t)}=W^{(t)}/\tr{W^{(t)}}$ and let $\Pi^{(t)}$ be the
projection onto the positive eigenspace of $S(\rho^{(t)})$.
\item[(b)]
Let $M^{(t)}=(N(\Pi^{(t)})+ r \I_\X)/2r$ and update the weight
matrix as follows:
\[
   W^{(t+1)}=exp(-\varepsilon \sum_{\tau=1}^t M^{(\tau)})
\]
\end{mylist}

\item[3.]
Choose $\bar{\rho}=\frac{1}{T} \sum_{\tau=1}^T \rho^{(t)}$ and let
$\bar{\Pi}$ be the projection onto the positive eigenspace of
$S(\bar{\rho})$.  Return $(\bar{\rho},\bar{\Pi})$ as the approximate
equilibrium point and $\ip{S(\bar{\rho})}{\bar{\Pi}}$ as the
approximate equilibrium value.
\end{mylist}
\noindent\hrulefill \caption{An algorithm that computes the
approximate value and point to precision $\delta$. }
\label{fig:mmw_generic}
\end{figure}

\begin{theorem}\label{thm:mmw_generic_equilibrium}
Let the input size, denoted by $|x|$, be the size of any instance in
\ceq{\ref{eqn:negative_equilibrium}}. If for any $\Pi \in T$ the
$N(\Pi)$ defined above satisfies $\snorm{N(\Pi)}_\infty \leq r$
where $r=O(polylog(|x|))$ and $\Psi(\rho)$ can be calculated in
\class{NC} for any $\rho \in \density{\X}$, then by using the
algorithm in Figure~\ref{fig:mmw_generic}, we can approximate the
equilibrium value to precision $\delta=\Omega(polylog(|x|))$ in
\class{NC}.
\end{theorem}

We will leave the proof of Theorem~\ref{thm:mmw_generic_equilibrium}
in Appendix~\ref{sec:MMW_survey}.


\section{QIP(2) Case} \label{sec:qip2}
Now it is our turn to consider a real instance of semidefinite
program and apply our framework to solve it. Our first candidate is
the quantum interactive proof system with two-messages. In this
system, after the input $x$ is given, the polynomial-time bounded
quantum verifier will send one quantum message to an all powerful
quantum prover and get another quantum message back. Then the
verifier will decide whether to accept or to reject based on the
message sent back from the prover and the qubits kept at his side.
The only constraint on the all powerful quantum prover is that the
prover must operate an admissable quantum operation on the quantum
message sent to him. The complexity class \class{QIP(2)} denotes all
the languages which can be recognized by the procedure above.
Precisely, we have

\begin{definition} \label{def:qip2}
Any language $L$ is inside \class{QIP(2)} if and only if
\begin{itemize}
  \item If $x \in L$, there exists a prover such that the verifier will \emph{accept} with
  probability at least $c(|x|)$.
  \item If $x \notin L$, for any prover the verifier will \emph{accept} with
  probability at most $s(|x|)$.
\end{itemize}
where $c(|x|)-s(|x|)=\Omega(1/poly(|x|))$.
\end{definition}

It is known that \class{QIP(2)}$\subseteq$\class{PSPACE}
~\cite{JainUW09} by following Kale's way~\cite{Kale07} to solve
SDPs. By contrast, we will demonstrate here how our
\textbf{Framework~\ref{framework_1}} can be applied to this problem.
Namely, we provide an alternative proof of the result
\class{QIP(2)}$\subseteq$\class{PSPACE}. The main difference between
our approach and the previous approach is that we formulate the
problem using the density operators instead of quantum channels and
we solve the semidefinite program using the new framework.

Let $\M$ denote the message's space between the prover and the
verifier and $\V$ denote the verifier's private space. Let us assume
the input $x$ is fixed for the following discussion. Without lost of
generality, let the pure state $\rho_1 \in \density{\M \otimes \V}$
be the initial state for the input $x$, namely the state that the
verifier prepares given input $x$. The prover will then operate an
admissable quantum channel $\Phi : \lin{\M} \rightarrow \lin{\M}$ on
part of the state $\rho_1$ and it will result another state
$\rho_2=\Phi \otimes \I_{\lin{\V}}(\rho_1)$. The verifier will
performance a POVM measurement on $\rho_2$ to decide whether to
accept or to reject. Let $R$ be the POVM which corresponds to the
case where the verifier accepts. In order to decide whether $x \in
L$,  it suffices to solve the optimization problem
\[
  \max_{\Phi} \ip{R}{\rho_2} \text{ s.t. } \rho_2=\Phi \otimes \I_{\lin{\V}}(\rho_1)
\]
where the optimum value is the maximum probability that the verifier
accepts given the input $x$.

Because of Lemma~\ref{lemma:admissable_connect}, we have $\rho_1$
and $\rho_2$ are connected by an admissable quantum operation if and
only if $\tr_{\M} \rho_1=\tr_{\M} \rho_2$. Thus the above
optimization problem is equivalent to the following SDP, denoted by
SDP (I).

\begin{center}
  \begin{minipage}{2in}
    \centerline{\underline{SDP Problem}}\vspace{-7mm}
    \begin{align*}
      \text{maximize:}\quad & \ip{R}{\rho_2}\\
      \text{subject to:}\quad & \tr_{\M}(\rho_2) \leq \tr_{\M}(\rho_1),\\
      & \rho_2\in \density{\M \otimes \V}.
    \end{align*}
  \end{minipage}
  \hspace*{12mm}
  \begin{minipage}{2in}
    \centerline{\underline{Feasibility Problem}}\vspace{-7mm}
    \begin{align*}
      \text{ask whether:}\quad & \ip{R}{\rho_2} \geq c\\
      \text{subject to:}\quad & \tr_{\M}(\rho_2) \leq \tr_{\M}(\rho_1),\\
      & \rho_2\in \density{\M \otimes \V}.
    \end{align*}
  \end{minipage}
\end{center}

\subsection{Solution to the Feasibility Problem}
Following the Framework~\ref{framework_1}, we consider the
feasibility problem above. Precisely, we define

\begin{equation} \label{eqn:h_1}
 f_1(\rho, \Pi)=\ip{\left(
                    \begin{array}{cc}
                       c- \ip{R}{\rho}  &  \\
                        & \tr_{\M}(\rho)-\tr_{\M}(\rho_1) \\
                    \end{array}
                  \right)}{\Pi}
\end{equation}
where $\rho \in T_1= \density{\M \otimes \V}$ and $\Pi \in
T_2=\{\Pi: 0\leq \Pi \leq \I_{\M \oplus \mathbb{C}}\}$. Let
$\equil{\lambda_1}$ be the equilibrium value of function $f_1$,
namely,
\[
  \equil{\lambda_1} =\min_{\rho \in T_1} \max_{\Pi \in T_2} f_1(\rho, \Pi) =\max_{\Pi \in T_2} \min_{\rho \in T_1} f_1(\rho,
  \Pi)
\]
Base on Theorem~\ref{thm:equilibrum_feasibility}, the value of
$\equil{\lambda_1}$ will imply whether the original problem is
feasible. In addition to that, we will demonstrate how to convert
any approximately feasible solution to exactly feasible solution
without changing the value of the object function a lot.

\begin{lemma} \label{lemma:promise_gap_qip2}
Assume we can calculate the approximate equilibrium value and point
of the function $f_1$ to the precision $\delta$. Let the
$\bar{\lambda_1}$ and $\{\bar{\rho}, \bar{\Pi}\}$ be the approximate
equilibrium value and point returned by algorithm in
Figure~\ref{fig:mmw_generic}, $\equil{\lambda_1}$ be the actual
equilibrium value. Then we have
\begin{itemize}
  \item if $\bar{\lambda_1}>\delta$, then the original problem is
  infeasible.
  \item if $\bar{\lambda_1}\leq \delta$, then there exists a
  feasible solution $\tilde{\rho}$ such that
  $\ip{R}{\tilde{\rho}}\geq c- \sqrt{2\delta-\delta^2}$.
\end{itemize}
\end{lemma}

\begin{proof}
\begin{itemize}
  \item If $\bar{\lambda_1}>\delta$, namely, $\equil{\lambda_1}\geq
  \bar{\lambda_1}-\delta>0$, then due to Theorem~\ref{thm:equilibrum_feasibility}, the original
  problem is feasible.
  \item Otherwise, due to Lemma~\ref{lemma:trace_norm}, we have
  \begin{equation} \label{eqn:qip2_lemma1}
    h_1(\bar{\rho},\bar{P})=\max \{ c-\ip{R}{\bar{\rho}},0\}
    +\frac{1}{2}\snorm{\tr_\M (\bar{\rho})-\tr_\M(\rho_1)}_1 \leq
    \delta
  \end{equation}
  By Lemma~\ref{lemma:purification_fidelity} and~\ref{lemma:compute_purification}, we can compute $\tilde{\rho}$
  such that $\fid(\tr_\M
  (\bar{\rho}),\tr_\M(\rho_1))=\fid(\tilde{\rho}, \rho_1)$ and
  $\tr_\M(\tilde{\rho})=\tr_\M(\rho_1)$.
  Let $s=\frac{1}{2}\snorm{\tr_\M (\bar{\rho})-\tr_\M(\rho_1)}_1$ and
  $t=\frac{1}{2}\snorm{\tilde{\rho}-\bar{\rho}}_1$.
  Then we have
  \begin{eqnarray*}
   \ip{R}{\tilde{\rho}}-c & = & \ip{R}{\bar{\rho}}+
   \ip{R}{\tilde{\rho}-\bar{\rho}} -c \\
   & \geq & \frac{1}{2}\snorm{\tr_\M
   (\bar{\rho})-\tr_\M(\rho_1)}_1-\delta -
   \frac{1}{2}\snorm{\tilde{\rho}-\bar{\rho}}_1 \\
   & = & s-t-\delta  \geq s-\sqrt{2s-s^2} -\delta \\
   & \geq & \delta-\sqrt{2\delta-\delta^2}-\delta =-\sqrt{2\delta-\delta^2}
  \end{eqnarray*}
where the first inequality is due to \ceq{\ref{eqn:qip2_lemma1}} and
Lemma~\ref{lemma:trace_norm} and the second inequality comes from
Lemma~\ref{lemma:trace_bound}. The last inequality is because
$s-\sqrt{2s-s^2}$ is decreasing when $0\leq s\leq 0.2$ and by
\ceq{\ref{eqn:qip2_lemma1}} $s\leq \delta$.
\end{itemize}
\end{proof}

\begin{theorem} \label{thm:qip_feasible}
For any guess $c$ ($ 0\leq c\leq 1$) for the feasibility problem,
and let the input size denoted by $|x|$ be the size of function
$f_1$, there is a \class{NC} algorithm to solve the feasibility
problem by returning either case in
Lemma~\ref{lemma:promise_gap_qip2}.
\end{theorem}

\begin{proof}
The algorithm basically follows the
Theorem~\ref{thm:mmw_generic_equilibrium} and
Lemma~\ref{lemma:promise_gap_qip2}. Since the particular feasibility
problem in our consideration has $\Psi(\cdot)=\tr_\M(\cdot)$ and
$A=R, B=\tr_\M(\rho_1)$, we have
\[
  N(\Pi)=-pR+\I_\M \otimes P +(pc- \ip{\tr_\M(\rho_1)}{P})\I_{\M
  \otimes \V}
\]
according to the definition in \ceq{\ref{eqn:N_def}} where $\Pi=\left(
                                                                  \begin{array}{cc}
                                                                    p &  \\
                                                                      & P \\
                                                                  \end{array}
                                                                \right)
                                                                \in
                                                                T_2$
in our problem. It is easy to see that
\[
  \snorm{N(\Pi)}_\infty \leq \snorm{-pR}_\infty+\snorm{\I_\M \otimes P}_\infty +\snorm{(pc- \ip{\tr_\M(\rho_1)}{P})\I_{\M
  \otimes \V}}_\infty \leq 1+1+1=3
\]
for any $\Pi \in T_2$ and $\Psi(\rho)=\tr_\M(\rho)$ can be
calculated in \class{NC} for any $\rho \in T_1$. Thus, by
Theorem~\ref{thm:mmw_generic_equilibrium}, we can compute the
approximate equilibrium value and point in
Lemma~\ref{lemma:promise_gap_qip2} to precision $\delta
=\Omega(1/polylog(|x|))$ in \class{NC}. Based on the two cases
discussed in Lemma~\ref{lemma:promise_gap_qip2}, we can either claim
the original problem is infeasible or calculate the $\tilde{\rho}$
by a \class{NC} algorithm according to
Lemma~\ref{lemma:compute_purification}. Compose all the \class{NC}
circuits above, then we have a \class{NC} algorithm for the
feasibility problem.
\end{proof}

\subsection{Solution to the Promised Version and General Case}

We are ready to apply the result of Theorem~\ref{thm:qip_feasible}
to simulate \class{QIP(2)} or more general cases. Recall the
definition of \class{QIP(2)}, there will be two promises with gap
$\Delta=c(|x|)-s(|x|)=\Omega(1/poly(|x|))$. Thus,

\begin{cor} \label{cor:qip2_pspace}
\class{QIP(2)} $\subseteq$ \class{PSPACE}
\end{cor}

\begin{proof}
For any input $x$, we simply compose the following circuits.
\begin{itemize}
  \item For any specific $x$, compute the corresponding initial
  state $\rho_1$ and the function $f_1$. This can be done in
  \class{NC(poly)} because it only involves the computation of the
  product of a polynomial number of exponential-size matrices that
  corresponds to the quantum circuits used by the verifier.
  \item Choose the guess value $c=\frac{1}{2}(c(|x|)+s(|x|))$
  and precision $\delta=\frac{1}{18}\Delta^2$. Then use the \class{NC} algorithm implied by
  Theorem~\ref{thm:qip_feasible} to calculate the equilibrium value
  $\equil{\lambda_1}$ to the precision $\delta$.
  \item Finally, based on the two cases in
  Lemma~\ref{lemma:promise_gap_qip2}, we can claim either the
  optimum value of the SDP is less than $c$ or at least
  $c-\sqrt{2\delta-\delta^2}\geq c-\frac{1}{3}\Delta$. Then we are
  able to tell whether $x \in L$.
\end{itemize}
Due to the facts in Appendix~\ref{sec:NC_parellel}, all the circuits
can be composed in \class{NC(poly)}. Because of the fact
\class{NC(poly)}=\class{PSPACE} ~\cite{Borodin77}, we have
\class{QIP(2)}$\subseteq$\class{PSPACE}.
\end{proof}

Furthermore, for the case where no such promise exists we can
develop a binary search to approximate the optimum value to high
precision.

\begin{theorem} \label{thm:qip2_binary_search}
Let $x$ be any instance of SDP (I) and $\alpha$ be the optimum value
of SDP (I). There exists a \class{NC} algorithm which can calculate
$\alpha$ to precision $\delta=\Omega(1/polylog(|x|))$. Furthermore,
there is a \class{NC} algorithm to compute $\rho_2$ such that
$\ip{R}{\rho_2}\geq \alpha-\delta$.
\end{theorem}

\begin{proof}
The proof follows from the binary search based on
Lemma~\ref{lemma:promise_gap_qip2}. Start with guess value $c$, by
Theorem~\ref{thm:qip_feasible}, we have a \class{NC} algorithm to
claim either $\alpha <c $ or $\alpha \geq c -\delta$. Thus, by using
binary search, we can calculate $\alpha$ to precision $\delta$.
Since $\delta=\Omega(1/polylog(|x|))$, there will be at most
polynomial-logarithm iterations in the binary search. Therefore, all
the circuits above can be composed in \class{NC}.
\end{proof}

\section{One-round Product QRG} \label{sec:qrg2}

In this section, we demonstrate how the Framework~\ref{framework_2}
can be applied to real problems. Here we consider the simplified
version of quantum refereed game with one round (two turns). The
upper bound of the complexity class recognized by the latter model,
denoted by \class{QRG(2)}, becomes more and more interesting after
the proof \class{QIP}=\class{PSPACE}~\cite{JainJUW09, Wu10}.
Particular, it is interesting to see whether
\class{QRG(2)}=\class{PSPACE} while its classical counterpart
\class{RG(2)} equals \class{PSPACE}~\cite{FeigeK97}.

The general one-round quantum refereed game works as follows. After
receiving some input $x$, the verifier then prepares some quantum
messages and send them to both Yes and No provers. After both
provers reply with quantum messages, the verifier will base on all
the quantum states at his hand to decide whether to accept $x$ or
not. Now, let us consider a simplified case where the messages sent
to the Yes prover and No prover are product states. We denote all
the languages recognized by this procedure by
\class{product-QRG(2)}. At the first sight, this might seem to be a
very restricted complexity class. However, by using the techniques
from the recent result~\cite{BeigiSW10}, we can prove the
\class{product-QRG(2)} contains all the languages which can
recognized by the most general model of one-round quantum refereed
game except the message sent to Yes prover is only of poly-logarithm
size.

Let us formulate the product one-round quantum refereed game in the
following way. Since the messages sent to both provers are product
states, let $\V_{Y}$ denote the verifier's private space when
interacts with the Yes prover and $\Y$ denote the message space
between the verifier and the Yes prover. Similarly, let $\V_{N}$
denote the verifier's private space when interacts with the No
prover and $\N$ denote the message space between the verifier and
the No prover. Without lose of generality, we can assume the pure
states $\rho_{Y} \in \density{\V_Y \otimes \Y}, \rho_{N} \in
\density{\V_N \otimes \N}$ are the initial states of the verifier
given some input $x$. We will assume the input $x$ is fixed in the
following discussion. Then the Yes and No provers will apply some
admissable quantum operations $\Phi_{Yes}, \Phi_{No}$ respectively
on part of the density operators $\rho_Y, \rho_N$. The resultant
states will be $\sigma_Y=\Phi_{Yes}\otimes \I_{\lin{\V_Y}} (\rho_Y),
\sigma_N=\Phi_{No}\otimes \I_{\lin{\V_N}} (\rho_N)$. Finally, the
verifier will make some POVM measurement on $\sigma_Y \otimes
\sigma_N$ to decide whether to accept or reject. Let $R$ ($0\leq
R\leq \I$) be the POVM which corresponds to the case where Yes
prover wins, then \emph{Game Value} $GV(R)$ is defined to be
\[
  GV(R) = \max_{\sigma_Y \in T_1} \min_{\sigma_N \in T_2} \ip{R}{\sigma_Y \otimes
  \sigma_N}= \min_{\sigma_N \in T_2} \max_{\sigma_Y \in T_1} \ip{R}{\sigma_Y \otimes
  \sigma_N}
\]
where
\[
T_1=\{\sigma \in \density{\V_Y \otimes \Y}  : \exists \text{
admissable } \Phi_{Yes}: \lin{\Y} \rightarrow \lin{\Y},
\sigma=\Phi_{Yes}\otimes \I_{\lin{\V_Y}} (\rho_Y)\}
\]
and
\[
T_2=\{\sigma \in \density{\V_N \otimes \N}  : \exists \text{
admissable } \Phi_{No}: \lin{\N} \rightarrow \lin{\N},
\sigma=\Phi_{No}\otimes \I_{\lin{\V_N}} (\rho_N)\}
\]
Since $\rho_Y, \rho_N$ are pure states, by
Lemma~\ref{lemma:admissable_connect}, we can simplify the definition
of $T_1, T_2$ to
\[
  T_1=\{\sigma_Y \in \density{\V_Y \otimes \Y}:
  \tr_\Y(\sigma_Y)=\tr_\Y(\rho_Y)\} \text{ , }  T_2=\{\sigma_N \in \density{\V_N \otimes \N}:
  \tr_\N(\sigma_N)=\tr_\N(\rho_N)\}
\]

\begin{definition} \label{def:product_QRG(2)}
Any language $L$ is inside this \class{product-QRG(2)} if and only
if for any input $x$,
\begin{itemize}
  \item If $x \in L$, then $GV(R)\geq c(|x|)$
  \item If $x \notin L$, then $GV(R) \leq s(|x|)$
\end{itemize}
where $c(|x|)-s(|x|)=\Omega(1/poly(|x|))$.
\end{definition}

Let $b=\frac{1}{2}(c(|x|)+s(|x|))$ and $\Delta=c(|x|)-s(|x|)$. By
applying Framework~\ref{framework_2}, we define the convex-concave
function $h_1$ as follows.
\[
  h_1(\{\sigma_Y,\Pi\}, \sigma_N)=\ip{R}{\sigma_Y \otimes
  \sigma_N}-b + \frac{2}{\Delta}\ip{\Pi}{\tr_\N(\sigma_N)-\tr_\N({\rho_N}) }
\]
where $\Pi$ comes from the set $T=\{\Pi: 0\leq \Pi \leq
\I_{\V_N}\}$. The function $h_1$ is actually the weighted sum where
the factor $\alpha$ is chosen to be $2/\Delta$ . Then we will
consider the new equilibrium value of function $h_1$ instead.
Precisely, we define
\[
  \equil{\mu} = \max_{\{\sigma_Y,\Pi\} \in T_1\times T} \min_{\sigma_N \in
  \density{\V_N \otimes N}} h_1(\{\sigma_Y,\Pi\}, \sigma_N)= \min_{\sigma_N \in
  \density{\V_N \otimes N}}\max_{\{\sigma_Y,\Pi\} \in T_1\times T} h_1(\{\sigma_Y,\Pi\}, \sigma_N)
\]
Then we are ready to show the gap between two promises in
Definition~\ref{def:product_QRG(2)} can be transferred to the
equilibrium value $\equil{\mu}$.

\begin{theorem} \label{lemma:promise_gap_qrg2}
Given the two promises in the Definition~\ref{def:product_QRG(2)},
we have for any input $x$,
\begin{itemize}
  \item If $x \in L$, then $\equil{\mu} \geq \frac{1}{4}\Delta$
  \item If $x \notin L$, then $\equil{\mu} \leq -\frac{1}{2}\Delta$
\end{itemize}
\end{theorem}

\begin{proof}
\begin{itemize}
  \item If $x \in L$, choose $(\{\equil{\sigma_Y}, \equil{\Pi}\}, \equil{\sigma_N})$ to be the equilibrium point. By Lemma
  ~\ref{lemma:trace_norm}, we have
  \[
    \equil{\mu} = \ip{R}{\equil{\sigma_Y} \otimes
\equil{\sigma_N}} -b+\frac{1}{\Delta}
\|\tr_\N(\equil{\sigma_N})-\tr_\N({\rho_N})\|_1
  \] By Lemma~\ref{lemma:trace_bound}, there exists some $\tilde{\sigma_N}$
  such that $\tr_\N(\tilde{\sigma_N})=\tr_\N({\rho_N})$ and
  $\fid(\tilde{\sigma_N},\equil{\sigma_N}
)=\fid(\tr_\N(\equil{\sigma_N}),\tr_\N({\rho_N}))$.
  Let $s=\frac{1}{2}\snorm{\tilde{\sigma_N}-
  \equil{\sigma_N}}_1$ and $t=\frac{1}{2}\snorm{\tr_\N(\equil{\sigma_N})-\tr_\N({\rho_N})}_1$ , thus
  \begin{eqnarray*}
    \equil{\mu} &=& \ip{R}{\equil{\sigma_Y} \otimes
\tilde{\sigma_N}} -b + \ip{R}{\equil{\sigma_Y} \otimes
(\equil{\sigma_N}-\tilde{\sigma_N})} + \frac{1}{\Delta}
\|\tr_\N(\equil{\sigma_N})-\tr_\N({\rho_N})\|_1 \\
      & \geq & \frac{1}{2} \Delta- \frac{1}{2}
\|\equil{\sigma_N}-\tilde{\sigma_N}\|_1 +\frac{1}{\Delta}
\|\tr_\N(\equil{\sigma_N})-\tr_\N({\rho_N})\|_1= \frac{\Delta}{2}  -s +\frac{2}{\Delta} t \\
      & \geq & \frac{\Delta}{2}-s+\frac{2}{\Delta}(1-\sqrt{1-s^2}) \geq \frac{\Delta}{2}+\frac{2}{\Delta}-\sqrt{1+\frac{4}{\Delta^2}} \\
  & = & \frac{\Delta}{2}+\frac{1-\sqrt{1+(\frac{\Delta}{2})^2}}{\frac{\Delta}{2}} \geq
  \frac{1}{4}\Delta
  \end{eqnarray*}
  where the first inequality comes from Lemma~\ref{lemma:trace_norm}
  and the second inequality is due to Lemma~\ref{lemma:trace_bound}. The third inequality is due to the
  fact that $\min_{0\leq s\leq 1}
  \frac{1}{x}-s-\frac{1}{x}\sqrt{1-s^2}=\frac{1}{x}-\sqrt{1+\frac{1}{x^2}}$.
  The last inequality comes from the fact $1-\sqrt{1+x^2}\geq
  -\frac{1}{2}x^2$ for any $0\leq x \leq 1$.

  \item If $x \notin L$, by Definition~\ref{def:product_QRG(2)} , there exists a $\tilde{\sigma_N}$ satisfying
  $\tr_\N(\tilde{\sigma_N})\leq\tr_\N({\rho_N})$
  and for any $\sigma_Y \in T_1$, we have $\ip{R}{\sigma_Y \otimes
  \tilde{\sigma_N}}-b \leq -\frac{\Delta}{2}$. Thus,
  \[
    \equil{\mu} \leq \max_{\{\sigma_Y,\Pi\} \in S_1\times T} h_1(\{\sigma_Y,\Pi\},
    \tilde{\sigma_N}) \leq -\frac{\Delta}{2}
  \]
\end{itemize}

\end{proof}

Thus in order to tell whether $x$ is inside $L$, it suffices to
compute the value of $\equil{\mu}$ and make the decision according
to Theorem~\ref{lemma:promise_gap_qrg2}. In the following proof, we
will implicitly make use of the result in
Theorem~\ref{thm:qip2_binary_search} recursively at each iteration
as the solution to the oracle. 

\begin{cor} \label{cor:qrg2}
\class{product-QRG(2)}$\subseteq$ \class{PSPACE}
\end{cor}

\begin{proof}
The proof of this corollary is quite similar to
Corollary~\ref{cor:qip2_pspace}. Whenever an input $x$ is given, we
can calculate $R$ in \class{NC(poly)} and approximate the
equilibrium value $\equil{\mu}$. By composing all these circuits, we
prove the whole circuit is in \class{NC(poly)} and thus in
\class{PSPACE}.

The only difference is the function $h_1$ is not in the same form as
we discussed in Theorem~\ref{thm:mmw_generic_equilibrium}. However,
by applying the general framework for the equilibrium value (see
Appendix~\ref{sec:MMW_survey}), we are able to calculate equilibrium
value $\equil{\mu}$ to high precision in \class{NC} as well.
Particularly, for each $\sigma^{(t)}_N$ generated, we will choose
$\sigma^{(t)}_Y$ to be the return value by
Theorem~\ref{thm:qip2_binary_search} where the POVM for SDP (I) is
$\tr_{\V_N \otimes N} ((\I_{\V_Y \otimes Y} \otimes
\sqrt{\sigma^{(t)}_N})R(\I_{\V_Y \otimes Y} \otimes
\sqrt{\sigma^{(t)}_N}))$. Furthermore, we will choose $\Pi^{(t)}$ to
be the projection onto the positive eignspace of
$\tr_\N(\sigma^{(t)}_N)-\tr_\N({\rho_N})$.  Similarly to the
algorithm in Figure~\ref{fig:mmw_generic}, we will choose
\[
  N(\sigma_Y, \Pi)=\tr_{\V_Y \otimes Y} ((\sqrt{\sigma_Y} \otimes \I_{\V_N \otimes N})R(\sqrt{\sigma_Y} \otimes \I_{\V_N \otimes
  N}))
  + \frac{2}{\Delta} \Pi \otimes \I_\N
  -(\frac{2}{\Delta}\ip{\Pi}{\tr_\N(\rho_N)}+b)\I_{\V_N \otimes N}
\]
Since $\Delta=\Omega(1/poly(|x|))$, we have $\snorm{N(\sigma_Y,
\Pi)}_\infty$ is always bounded by $O(poly(|x|))$. Thus the whole
algorithm can be accomplished in \class{NC(poly)}.

\end{proof}

\section{Conclusions and Open Problems} \label{sec:conclusion}
In this paper we demonstrate how the Framework~\ref{framework_1} can
be used to solve a class of SDPs. Moreover, by the examples of
\class{QIP(2)} and \class{QMAM}, we demonstrate how the conversion
from approximate to exact feasibility can be done in our framework.
The generic form in Theorem~\ref{thm:mmw_generic_equilibrium} also
illustrates the potential of our framework to solve other SDPs. In
addition, our example of \class{product-QRG(2)} illustrates how the
Framework~\ref{framework_2} can be used to calculate the equilibrium
value of more complicated form.

However, there are several limits and unknown facts about our two
frameworks. As mentioned in~\cite{JainJUW09}, it might be impossible
to solve any SDPs in \class{NC}. Thus, we cannot hope to include all
possible SDPs into our framework. Understanding what kind of
constraints for SDPs can be solved via our framework is a major open
problem. So far, we have positive results when the constraints are
partial trace and its simple combination. It can be easily verified
that the constraint $\Psi=\tr_\A(U \cdot U^\ast)$ for some unitary
$U$ can also be solved in the similar way.

Another open problem is whether any generalization of the
multiplicative weight update method can be found to improve the
results based on the multiplicative weight update method. The recent
survey paper~\cite{Hazan10} could provide some insights into that.

Finally, it is still open whether \class{QRG(2)}=\class{PSPACE}
while its classical counterpart \class{RG(2)}=\class{PSPACE}.

\section*{Acknowledgement}
The author is grateful to Gus Gutoski, Zhengfeng Ji, Rahul Jain,
Richard Cleve, Sarvagya Upadhyay and John Watrous for helpful
discussions. Particularly, the penalization idea in
Framework~\ref{framework_2} is inspired by discussion with Gus
Gutoski and the weight factor in the penalization is suggested by
John Watrous. The fact that the constraint $\Psi=\tr_\A(U \cdot
U^\ast)$ for some unitary $U$ can also be solved in the similar way
is noticed by Zhengfeng Ji. The author is grateful to Gus Gutoski
and Zhengfeng Ji for helpful comments on the manuscript. The author
also wants to thank the hospitality and invaluable guidance of John
Watrous when he was visiting the Institute for Quantum Computing,
University of Waterloo. The research was completed during this visit
and was supported by the Canadian Institute for Advanced Research
(CIFAR).

\bibliographystyle{alpha}

\begin{thebibliography}{BOGKW88}

\bibitem[Ali95]{Alizadeh95}
F.~Alizadeh.
\newblock Interior point methods in semidefinite programming with applications
  to combinatorial optimization.
\newblock {\em SIAM Journal on Optimization}, 5(1):13--51, 1995.

\bibitem[AHK05a]{AroraHK05a}
S.~Arora, E.~Hazan, and S.~Kale.
\newblock Fast algorithms for approximate semidefinite programming using the
  multiplicative weights update method.
\newblock In {\em Proceedings of the 46th Annual IEEE Symposium on Foundations
  of Computer Science}, pages 339--348, 2005.

\bibitem[AHK05b]{AroraHK05b}
S.~Arora, E.~Hazan, and S.~Kale.
\newblock The multiplicative weights update method: a meta algorithm and applications.
2005

\bibitem[AK07]{AroraK07}
S.~Arora and S.~Kale.
\newblock A combinatorial, primal-dual approach to semidefinite programs.
\newblock In {\em Proceedings of the Thirty-Ninth Annual ACM Symposium on
  Theory of Computing}, pages 227--236, 2007.

\bibitem[Bab85]{Babai85}
L.~Babai.
\newblock Trading group theory for randomness.
\newblock In {\em Proceedings of the 17th Annual ACM Symposium on Theory of
  Computing}, pages 421--429, 1985.

\bibitem[BSW10]{BeigiSW10}
S.~Beigi, P.W.~Shor and J.~Watrous
\newblock Quantum interative proofs with short messages
\newblock availabe at arXiv:1004.0411, 2010.

\bibitem[BCP83]{BorodinCP83}
A.~Borodin, S.~Cook, and N.~Pippenger.
\newblock Parallel computation for well-endowed rings and space-bounded
  probabilistic machines.
\newblock {\em Information and Control}, 58:113--136, 1983.

\bibitem[BGH82]{BorodinGH82}
A.~Borodin, J.~von~zur Gathen, and J.~Hopcroft.
\newblock Fast parallel matrix and {GCD} computations.
\newblock In {\em Proceedings of the 23rd Annual IEEE Symposium on Foundations
  of Computer Science}, pages 65--71, 1982.

\bibitem[Bha97]{Bhatia97}
R.~Bhatia.
\newblock {\em Matrix Analysis}.
\newblock Springer, 1997.

\bibitem[Bor77]{Borodin77}
A.~Borodin.
\newblock On relating time and space to size and depth.
\newblock {\em SIAM Journal on Computing}, 6:733--744, 1977.

\bibitem[BV04]{BoydV04}
S.~Boyd and L.~Vandenberghe
\newblock {\em Convex Optimization}.
\newblock Cambridge University Press, 2004

\bibitem[dK02]{deKlerk02}
E.~de~Klerk.
\newblock {\em Aspects of Semidefinite Programming -- Interior Point Algorithms
  and Selected Applications}, volume~65 of {\em Applied Optimization}.
\newblock Kluwer Academic Publishers, Dordrecht, 2002.

\bibitem[FK97]{FeigeK97}
U.~Feige and J.~Kilian.
\newblock Making games short.
\newblock In {\em Proceedings of the 29th Annual ACM Symposium on Theory of
  Computing}, pages 506--516, 1997.

\bibitem[Fle00]{Fle00}
L.K.~Fleischer.
\newblock Approximating fractional multicommodity flow independent of the
number of commodities.
\newblock { \em SIAM J. Discret. Math.},
13(4):505¨C520, 2000

\bibitem[FS99]{FreundS99}
Y.~Freund and R.E.~Schapire.
\newblock Adaptive game playing using
multiplicative weights.
\newblock {\em Games and Economic Behavior  },
29:79--103, 1999.

\bibitem[Haz06]{Hazan06}
E.~Hazan
\newblock Approximate Convex Optimization by Online Game Playing.
\newblock available at arxiv cs.CC/0610119, 2006.

\bibitem[Haz10]{Hazan10}
E.~Hazan
\newblock  The convex optimization approach to regret minimization.
\newblock  manuscript, 2010.

\bibitem[Gat93]{vzGathen93}
J.~von~zur Gathen.
\newblock Parallel linear algebra.
\newblock In J.~Reif, editor, {\em Synthesis of Parallel Algorithms},
  chapter~13. Morgan Kaufmann Publishers, Inc., 1993.

\bibitem[GK98]{GargK98}
N.~Garg and J.~K\"{o}nemann.
\newblock Faster and simpler algorithms
for multicommodity flow and other fractional packing problems.
\newblock In {\em Proceedings of the 39th Annual Symposium on Foundations of
Computer Science}, pages 300--309, 1998

\bibitem[GMR85]{GoldwasserMR85}
S.~Goldwasser, S.~Micali, and C.~Rackoff.
\newblock The knowledge complexity of interactive proof systems.
\newblock In {\em Proceedings of the 17th Annual ACM Symposium on Theory of
  Computing}, pages 291--304, 1985.

\bibitem[Gut05]{Gutoski05}
G.~Gutoski
\newblock Upper bounds for quantum interactive proofs with competing
provers.
\newblock In {\em Proceedings of the 20th Annual IEEE Conference on Computational Complexity}, pages 334--343, 2005.

\bibitem[GW05]{GutoskiW05}
G.~Gutoski and J.~Watrous
\newblock Quantum interactive proofs with competing
provers.
\newblock In {\em Proceedings of the 22th Symposium on Theoretical Aspects of Computer Science},
  volume 3404 of {\em Lecture Notes in Computer Science}, pages 605--616. Springer 2005.

\bibitem[GW07]{GutoskiW07}
G.~Gutoski and J.~Watrous
\newblock Toward a general theory of quantum games.
\newblock In {\em Proceedings of the 39th ACM Symposium on Theory of Computing}, pages 565--574, 2007.

\bibitem[JJUW09]{JainJUW09}
R. Jain, Z. Ji, S. Upadhyay, and J. Watrous.
\newblock QIP = PSPACE.
\newblock In {\em Proceedings of the 42nd ACM Symposium on Theory of Computing}, 2010

\bibitem[JUW09]{JainUW09}
R.~Jain, S.~Upadhyay, and J.~Watrous.
\newblock Two-message quantum interactive proofs are in {PSPACE}.
\newblock In {\em Proceedings of the 50th Annual IEEE Symposium on Foundations
  of Computer Science}, pages 534--543, 2009.

\bibitem[JW09]{JainW09}
R.~Jain and J.~Watrous.
\newblock Parallel approximation of non-interactive zero-sum quantum games.
\newblock In {\em Proceedings of the 24th IEEE Conference on Computational
  Complexity}, pages 243--253, 2009.

\bibitem[Kal07]{Kale07}
S.~Kale.
\newblock {\em Efficient algorithms using the multiplicative weights update
  method}.
\newblock PhD thesis, Princeton University, 2007.

\bibitem[Kha04]{Khandekar04}
R.~Khandekar
\newblock {\em Lagrangian Relaxation based Algorithms
for Convex Programming Problems}.
\newblock PhD thesis, Indian Institute of Technology, Delhi, 2004

\bibitem[KW00]{KitaevW00}
A.~Kitaev and J.~Watrous.
\newblock Parallelization, amplification, and exponential time simulation of
  quantum interactive proof system.
\newblock In {\em Proceedings of the 32nd Annual ACM Symposium on Theory of
  Computing}, pages 608--617, 2000.

\bibitem[KSV02]{KitaevW02}
 A.~Kitaev, A.~Shen, M.~Vyalyi
\newblock Classical and Quantum Computation
\newblock  American Mathematical Society, 2002

\bibitem[Lov03]{Lovasz03}
L.~Lovasz.
\newblock Semidefinite programs and combinatorial optimization.
\newblock {\em Recent Advances in Algorithms and Combinatorics}, 2003.

\bibitem[NC00]{NielsenC00}
M.~A. Nielsen and I.~L. Chuang.
\newblock {\em Quantum Computation and Quantum Information}.
\newblock Cambridge University Press, 2000.

\bibitem[PST91]{PlotkinST91}
S.A.~Plotkin, D.B.~Shmoys, and T.~Tardos.
\newblock Fast approximation algorithm for fractional packing and covering
problems.
\newblock In {\em Proceedings of the 32nd Annual IEEE Symposium on Foundations
of Computer Science}, pages 495--504, 1991.

\bibitem[Ste10]{Steurer10}
David Steurer.
\newblock Fast SDP Algorithms for Constraint Satisfaction Problems.
\newblock In {\em Proceedings of the 21st Annual ACM-SIAM Symposium on Discrete
Algorithms}, 2010.

\bibitem[Rag08]{Raghavendra08}
P.~Raghavendra.
\newblock Optimal algorithms and inapproximability results for every CSP?
\newblock In {\em Proceedings of the 40th annual ACM symposium on Theory of
computing}, pages 245--254, 2008.

\bibitem[vN28]{vN28}
J.~von Neumann.
\newblock Zur theorie der gesellschaftsspiele.
\newblock {\em Mathematische Annalen}, 100(1928):295--320, 1928

\bibitem[VB96]{VandenbergheB96}
L.~Vandenberghe and S.~Boyd.
\newblock Semidefinite programming.
\newblock {\em SIAM Review}, 38(1):49--95, 1996.

\bibitem[Vaz01]{Vazirani01}
V.~Vazirani.
\newblock {\em Approximation Algorithms }
\newblock Springer-Verlag, Berlin, 2001

\bibitem[Wat08]{Watrous08}
J.~Watrous.
\newblock Lecture Notes for Theory of Quantum Information
\newblock Fall 2008.

\bibitem[Wat09b]{Watrous09b}
J.~Watrous.
\newblock Semidefinite programs for completely bounded norms.
\newblock {\em Theory of Computing 5: 11}, 2009.

\bibitem[WK06]{WarmuthK06}
M.~Warmuth and D.~Kuzmin.
\newblock Online variance minimization.
\newblock In {\em Proceedings of the 19th Annual Conference on Learning
  Theory}, volume 4005 of {\em Lecture Notes in Computer Science}, pages
  514--528. Springer, 2006.

\bibitem[Wu10]{Wu10}
Xiaodi Wu.
\newblock Equilibrium Value Method for the Proof of QIP=PSPACE
\newblock availabe at arxiv: 1004.0264

\bibitem[You95]{Young95}
N.E.~Young
\newblock Randomized rounding without solving the linear
program.
\newblock In {\em Proceedings of the Sixth Annual ACM-SIAM Symposium on Discrete
Algorithms}, pages 170--178,  1995.

\end{thebibliography}

\appendix

\section{Extended Preliminaries  } \label{sec:appendix_prelim}

\subsection{Fundamentals of Quantum Information}
\label{sec:prelim_fundamentals_qi} In this section, we will provide
a summary of the fundamental notations and facts in quantum
information. We assume the readers are familiar with these
knowledge, and most part of this section is meant to make clear the
terminology and well-known facts used in this paper. For those
readers who are not familiar with these concepts, we recommend them
to refer to ~\cite{Bhatia97,KitaevW02,NielsenC00,Watrous08}. Our
notation basically follows the notation in Watrous's lecture
notes~\cite{Watrous08}.

A \emph{quantum register} refers to a collection of qubits, usually
represented by a complex Euclidean spaces of the form
$\X=\mathbb{C}^\Sigma$ where $\Sigma$ refers to some finite
non-empty set of the possible states.

For any two complex Euclidean spaces $\X,\Y$, let $\lin{\X,\Y}$
denote the space of all linear mappings(or operators) from $\X$ to
$\Y$ ($\lin{\X}$ short for $\lin{\X,\X}$). An operator $A \in
\lin{\X,\Y}$ is a \emph{linear isometry} if $A^{\ast}A=\I_\X$ where
$A^{\ast}$ denotes the adjoint(or conjugate transpose) of $A$. The
set of linear isometries is denoted by $\unitary{\X,\Y}$ then.

 An operator $A\in \lin{\X}$ is \emph{Hermitian}, the set of which is denoted by $\herm{\X}$, if
$A=A^{\ast}$. The eigenvalues of a Hermitian operator are always
real. For $n=\dim{\X}$, we write $ \lambda_1(A)\geq \lambda_2(A)\geq
\cdots \geq \lambda_n(A)$ to denote the eigenvalues of $A$ sorted
from largest to smallest. An operator $P\in \lin{\X}$ is
\emph{positive semidefinite}, the set of which is denoted by
$\pos{\X}$, if $P$ is Hermitian and all of its eigenvalues are
nonnegative, namely $\lambda_n(P)\geq 0$. An operator $\rho \in
\pos{\X}$ is a \emph{density operator}, the set of which is denoted
by $\density{\X}$,  if it has trace equal to 1. A density operator
$\rho \in \density{\X}$ is said to be \emph{pure} if it has rank
equal to one.  An operator $\Pi \in \pos{\X}$ is a \emph{projection}
if $\Pi$ projects onto some subspace of $\X$. Furthermore, such
operators only have eigenvalues of 0 or 1.

The Hilbert-Schmidt inner product on $\lin{\X}$ is defined by
\[
   \ip{A}{B}=\tr{A^{\ast}B}
\]
for all $A,B \in \lin{\X}$.

A \emph{super-operator}(or quantum channel) is a linear mapping of
the form
\[
  \Psi : \lin{\X} \rightarrow \lin{\Y}
\]
A super-operator $\Psi$ is said to be \emph{positive} if $\Psi(X)
\in \pos{\Y}$ for any choice of $X \in \pos{\X}$, and is
\emph{completely positive} if $\Psi \otimes \I_{\lin{\Z}}$ is
positive for any choice of a complex vector space $\Z$.  The
super-operator $\Psi$ is said to be \emph{trace-preserving} if
$\tr{\Psi(X)} =\tr{X}$ for all $X \in \lin{\X}$. A super-operator
$\Psi$ is \emph{admissable} if it is completely positive and
trace-preserving. Admissable super-operators represent the
discrete-time changes in quantum systems that, in principle, can be
physically realized.

 We refer to \emph{measurements}, or precisely POVM-type
measurements as a collection of positive semidefinite operators
\[
   \{P_a : a \in \Sigma\} \subset \pos{\X}
\]
satisfying the constraint $\sum_{a \in \Sigma} P_a =\I_\X$. Here
$\Sigma$ refers to a finite, nonempty set of \emph{measurement
outcomes}. If a quantum state represented by $\rho \in \density{\X}$
is measured with respect to this measurement, then each outcome $a
\in \Sigma$ will be observed with probability $\ip{P_a}{\rho}$.

The \emph{trace norm} of an operator $A \in \lin{\X}$ is denoted by
$\snorm{A}_1$ and defined to be
\[
  \snorm{A}_1= \tr \sqrt{ A^\ast A}
\]
When $A$ is Hermitian, we have
\begin{equation} \label{def:eqn_trace_ip}
  \snorm{A}_1= \max\{ \ip{P_0-P_1}{A}: P_0,P_1 \in \pos{\X},
  P_0+P_1=\I_\X\}
\end{equation}

 The \emph{spectral norm} of an operator $A \in \lin{\X}$ is
defined to be
\[
   \snorm{A}_\infty =\max \{ \snorm{Au}: u \in \X, \snorm{u}=1\}
\]

Given two positive semidefinite operators $P,Q \in \pos{\X}$, we
define the \emph{fidelity} between $P$ and $Q$ as
\[
  \fid(P, Q)=\snorm{\sqrt{P}\sqrt{Q}}_1
\]
When $P,Q$ are density operators, due to Fuchs-van de Graaf
inequality, we have
\begin{equation} \label{def:eqn_van_de_Graaf}
 1-\frac{1}{2}\snorm{P-Q}_1 \leq \fid(P,Q) \leq
 \sqrt{1-\frac{1}{4}\snorm{P-Q}_1^2}
\end{equation}

Suppose $\rho \in \density{\X}$ is a density operator, the
\emph{purification} of $\rho $ in $\X \otimes \Y$ is any pure
density operator $uu^\ast \in \density{\X \otimes \Y}$ for which
$\tr_\Y(uu^\ast)=\rho$.

The following contains the proof of lemmas shown in
Section~\ref{sec:lemma_facts}

\noindent \textbf{Proof of Lemma~\ref{lemma:purification_fidelity}}
(Please note that this lemma was originally proved in many places.
The following proof follows the one in~\cite{JainUW09}. The only
reason to include this proof is because we will use it to prove
Lemma~\ref{lemma:compute_purification}.)
\begin{proof}
First, by the monotonicity of the fidelity function under partial
trace, we have for any $\sigma_2 \in \density{\A \otimes \B}$ such
that $\tr_B \sigma_2=\rho_2$ the inequality
$\mathcal{F}(\rho_1,\rho_2)\geq \mathcal{F}(\sigma_1,\sigma_2)$
always holds. Thus, it suffices to show that equality can be
achieved.

Let $V \in \unitary{\A}$ such that $\sqrt{\rho_1}\sqrt{\rho_2}V$ is
positive semidefinite. Since for fidelity function we have
$\mathcal{F}(\rho_1,\rho_2)=\snorm{\sqrt{\rho_1}\sqrt{\rho_2}}_1$,
then for such a $V$ it holds that $\mathcal{F}(\rho_1,\rho_2)=\tr
(\sqrt{\rho_1}\sqrt{\rho_2}V)$. Now let $\C=\A \otimes \B$ and
$\ket{u_1} \in \A \otimes \B \otimes \C$ be the purification of
$\sigma_1$, in particular, $\ket{u_1}$ is chosen to be
\[
   \ket{u_1}= \vec ( \sqrt{\sigma_1})
\]
By rearranging the coefficients we can find a $X \in \lin{\B \otimes
\C, \A}$ such that $\vec(X)=\ket{u_1}$. Since $\ket{u_1}$ is also a
purification of $\rho_1$, there must exist a linear isometry $U \in
\unitary{\A, \B \otimes \C}$ such that,
\[
  X= \sqrt{\rho_1} U^\ast
\]

Finally, let $\ket{u_2}=\vec(\sqrt{\rho_2}VU^\ast) \in \A \otimes \B
\otimes \C$ where $V, U$ are obtained above respectively. It is easy
to see that $\ket{u_2}$ is a purification of $\rho_2$ in the space
$\A \otimes \B \otimes \C$. Thus, we choose
$\sigma_2=\tr_\C(\ket{u_2}\bra{u_2})$ and it will hold that
\[
  \fid(\sigma_1,\sigma_2) \geq
  |\ip{\vec(\sqrt{\rho_1}U^\ast)}{\vec(\sqrt{\rho_2}VU^\ast)}| =
  |\ip{\sqrt{\rho_1}U^\ast}{\sqrt{\rho_2}VU^\ast}|=\tr(\sqrt{\rho_1}\sqrt{\rho_2}V)=\fid(\rho_1,\rho_2)
\]
\end{proof}

\noindent \textbf{Proof of Lemma~\ref{lemma:compute_purification}}
\begin{proof}
The proof of the Lemma~\ref{lemma:purification_fidelity} actually
gives you a way to construct such a $\sigma_2$ given
$\rho_1,\rho_2,\sigma_1$. Let us review the important steps in the
proof again with more attention to the computation of each
intermediate quantity.

In the first step, we need to calculate a $V \in \unitary{\A}$ such
that $\sqrt{\rho_1}\sqrt{\rho_2}V$ is positive semidefinite. This
can be done by calculating the singular value decomposition of
$\sqrt{\rho_1}\sqrt{\rho_2}$ and let $V=\I-2P$ where $P$ is the
projection onto the subspace with negative singular values.

The second step calculates $X$ such that
$\vec(X)=\ket{u_1}=\vec(\sqrt{\sigma_1})$. This can be done by
simply rearranging the coefficients in the entries of
$\sqrt{\sigma_1}$. In order to get $U \in \unitary{\A ,\B \otimes
\C}$, we can calculate the singular value decomposition of
$\sqrt{\rho_1}$ and get the inverse( or pseudo-inverse) of
$\sqrt{\rho_1}$. Then $U=X^\ast (\sqrt{\rho_1}^{-1})^\ast$.

Once we have $U$ and $V$, we can easily calculate $\sigma_2$ by
using the formula
\[
\sigma_2=\tr_\C (\vec(\sqrt{\rho_2}VU^\ast)
\vec(\sqrt{\rho_2}VU^\ast)^\ast) \]
 Due to the fact that fundamental operations of matrix and the singular value decomposition can
be done in \class{NC}(see Fact~\ref{fact:fundamentals} and
Fact~\ref{fact:svd}) and the fact we can compose these \class{NC}
circuits easily, then we can conclude that $\sigma_2$ can be
calculated in \class{NC} given the classical representations of
$\rho_1,\rho_1$ and $\sigma_1$ as input.
\end{proof}

\noindent \textbf{Proof of Lemma~\ref{lemma:trace_bound}}
\begin{proof}
This is only a simple application of the Fuchs-van de Graaf
Inequalities(Eq[\ref{def:eqn_van_de_Graaf}]). Namely, we have
\[
  1-s \leq \fid(\rho_1, \rho_2) \leq \sqrt{1-s^2} \text{ and } 1-t
  \leq \fid(\sigma_1,\sigma_2) \leq \sqrt{1-t^2}
\]
Given the fact $\fid(\rho_1,\rho_2)=\fid(\sigma_1, \sigma_2)$, it
easily follows that
\[
  1-s\leq \sqrt{1-t^2} \text {  and } 1-t \leq \sqrt{1-s^2}
\]
\end{proof}

\noindent \textbf{Proof of Lemma~\ref{lemma:admissable_connect}}
\begin{proof}
This lemma is a standard fundamental fact of quantum information
that follows from the unitary equivalence of purifications and the
fact $\rho_1$ is a pure state.
\end{proof}

\noindent \textbf{Proof of Lemma~\ref{lemma:trace_norm}}

\begin{proof}
Given the fact about the trace norm of any Hermitian operator $A$ in
Equation~\ref{def:eqn_trace_ip}, we have
\begin{eqnarray*}
  \snorm{A}_1  &=& \max_{\Pi: 0 \leq \Pi \leq \I} \ip{A}{ \Pi -(\I-\Pi)} \\
    &=& \max_{\Pi: 0 \leq \Pi \leq \I} 2\ip{A}{\Pi}-\tr (A) \\
    &=& \max_{\Pi: 0 \leq \Pi \leq \I} 2\ip{A}{\Pi}
\end{eqnarray*}
Thus,
\[
  \max_{\Pi: 0\leq \Pi \leq \I} \ip{A}{\Pi}=\frac{1}{2}\snorm{A}_1
\]
For the other equation, choose $A'=-A$, and apply the equation
above, we have
\[
 \max_{\Pi: 0\leq \Pi \leq \I} \ip{A'}{\Pi}=\max_{\Pi: 0\leq \Pi \leq \I}
 \ip{-A}{\Pi}=\frac{1}{2}\snorm{A'}_1=\frac{1}{2}\snorm{A}_1
\]
Thus
\[
 \min_{\Pi: 0\leq \Pi \leq \I} \ip{A}{\Pi}=-\frac{1}{2}\snorm{A}_1
\]
\end{proof}

\subsection{Facts on NC and parallel matrix computations}
\label{sec:NC_parellel}

We denote by \class{NC} the class of promise problems computed by
the logarithmic-space uniform Boolean circuits with poly-logarithmic
depth. Furthermore, we denote by \class{NC(poly)} the class of
promise problems computed by the polynomial-space uniform Boolean
circuits with polynomial depth. Since it holds that
\class{NC(poly)}$=$ \class{PSPACE}~\cite{Borodin77}, thus in order
to simulate the algorithm above in PSPACE, it suffices to prove that
we can simulate the algorithm in \class{NC(poly)}.

There are a few facts about these classes which are useful in our
discussion. The first fact is the functions in these classes compose
nicely. It is clear that if $f \in$ \class{NC(poly)} and $g \in$
\class{NC}, then their composition $g\circ f$ is in
\class{NC(poly)}, which follows from the most obvious way of
composing the families of circuits. Another useful fact is that many
computations involving matrices can be performed by NC algorithms
(Please refer to the survey~\cite{vzGathen93} which describes NC
algorithms for these tasks). Especially, we will make use of the
fact that matrix exponentials and singular value decompositions can
be approximated to high precision in NC. We will directly cite the
well-prepared form of these facts in~\cite{JainJUW09}.

\begin{fact} \label{fact:fundamentals}
Fundamental operations like addition, multiplication of matrices can
be done in \class{NC}
\end{fact}

\begin{fact} \label{fact:matrix_exponential}
\underline{Matrix exponentials}: there exists \class{NC} algorithms
such that
\begin{center}
\begin{tabular}{lp{5.5in}}
{\it Input:} & An $n\times n$ matrix $M$, a positive rational number
$\eta$,
and an integer $k$ expressed in unary notation (i.e., $1^k$).\\
{\it Promise:} & $\norm{M} \leq k$.\\ {\it Output:} & An $n\times n$
matrix $X$ such that $\norm{\exp(M) - X} < \eta$.
\end{tabular}
\end{center}
\end{fact}

\begin{fact} \label{fact:svd}
\underline{Singular value decompositions}: there exists \class{NC}
algorithms such that
\begin{center}
\begin{tabular}{lp{5.5in}}
{\it Input:} & An $m\times n$ matrix $M$ and a positive rational
number $\eta$. \\ {\it Output:} & An $m\times m$ unitary matrix $U$
,$n\times n$ unitary matrix $V$ and an $m\times n$ real diagonal
matrix $\Lambda$ such that
\[
\norm{M - U \Lambda V^\ast} < \eta.
\]
\end{tabular}
\end{center}

\end{fact}

\subsection{Multiplicative Weights Update Method}
\label{sec:MMW_survey}

 The \emph{multiplicative weights update method} introduced in Section \ref{sec:introduction} is a framework
for algorithm design(or meta-algorithm) that works as the one shown
in Fig~\ref{fig:mmw}. This kind of framework involves lots of
technical details and we refer the curious reader to the survey and
the PhD thesis~\cite{Kale07} mentioned in the introduction. However,
for the sake of completeness, we provide the main result which will
be useful in our proof. It should be noticed that $\{M^{(t)}\}$ is
the freedom we have in this framework.

%

\begin{figure}[t]
\noindent\hrulefill
\begin{mylist}{8mm}
\item[1.] Initialization: Pick a fixed $\varepsilon\leq \frac{1}{2}$,
and let $W^{(1)}=\I_{\X} \in \lin{\X}$, $D=\dim{\X}$.
\item[2.]
Repeat for each $t = 1,\ldots,T$:
\begin{mylist}{8mm}
\item[(a)] Let the density operator
$\rho^{(t)}=W^{(t)}/\tr{W^{(t)}}$
\item[(b)]  Observe the loss matrix $M^{(t)} \in \lin{\X}$ which satisfies $-\I_\X \leq M^{(t)}\leq 0$ or $0\leq M^{(t)}\leq \I_\X$,
 update the weight matrix as follows:
\[
   W^{(t+1)}=exp(-\varepsilon \sum_{\tau=1}^t M^{(\tau)})
\]
\end{mylist}
\end{mylist}
\noindent\hrulefill \caption{The Matrix Multiplicative Weights
Update method.} \label{fig:mmw}
\end{figure}

\begin{theorem} \label{thm:mmw_simple}
Assume $0\leq M^{(t)} \leq \I$ for all $t$, after $T$ rounds, the
algorithm in Fig~\ref{fig:mmw} guarantees that, for any $\rho^{\ast}
\in \density{\X}$, we have
\begin{equation}\label{eqn:mmw_simple}
     (1-\epsilon)\sum_{t=1}^T
    \ip{\rho^{(t)}}{M^{(t)}} \leq \ip{\rho^{\ast}}{\sum_{t=1}^T
    M^{(t)}} + \frac{lnD}{\epsilon}
\end{equation}
\end{theorem}
The proof can be found in Kale's thesis~\cite{Kale07} or the
appendix of~\cite{Wu10}. We will then discuss how this method can
used to solve the feasibility problem and equilibrium value
following the way in Kale's thesis~\cite{Kale07} with more details.

In order to solve any feasibility problem of general form, the
primal-dual method will generate a series of candidate solutions
$X^{(1)}, X^{(2)}, \cdots, X^{(T)}$ for $T$ rounds. For any
$X^{(t)}$ in the round $t$, we require an oracle $\mathcal{O}_1$ to
solve the following problem
\begin{equation} \label{eqn:sdp_oracle_work}
 \text{find } Y^{(t)} \text{ s.t. }
 \ip{\Psi^\ast(Y^{(t)})-A}{X^{(t)}} \geq 0, \ip{B}{Y^{(t)}}\leq c,
 Y^{(t)} \in \pos{\Y}
\end{equation}
The oracle $\mathcal{O}_1$ will return such a $Y^{(t)}$ or claim
such a $Y^{(t)}$ does not exist. If such a $Y^{(t)}$ exists, the
primal-dual method will generate the $X^{(t+1)}$ via the
multiplicative weight update method by choosing $M^{(t)}$ to be a
renormalized version of $\Psi^\ast(Y^{(t)})-A$. Otherwise, it stops
to claim the rescaled $X^{(t)}$ is feasible to the primal problem
and $\ip{A}{X^{(t)}}\geq c$. If the method does not stop for $T$
rounds, the multiplicative weight update method will generate an
approximate dual feasible solution $Y$. Under certain conditions,
such a $Y$ can be converted into exact dual feasible solution
$\tilde{Y}$ such that $\ip{B}{\tilde{Y}} \leq (1+\epsilon)c$. Due to
the duality of SDPs, this implies that $\alpha\leq \beta\leq
(1+\epsilon)c$.

In addition to the convertibility from approximate to exact
feasibility, another difficulty in applying this generic method is
the design of the oracle $\mathcal{O}_1$ . Efficient solution to the
oracle $\mathcal{O}_1$ is necessary to guarantee an efficient
algorithm for the feasibility problem. We will refer this
requirement as the \emph{efficient solvability}. In addition, we
need the spectrum of $\Psi^\ast(Y^{(t)})-A$ of the oracle
$\mathcal{O}_1$ is bounded within a small range. We will refer this
requirement as \emph{width-boundness}.

The generic framework to calculate the equilibrium value is
different in the sense of the design of the oracle and the use of
Theorem~\ref{thm:mmw_simple}. Consider the value $\lambda$,
\[
  \lambda=  \min_{x \in X} \max_{y \in Y} f(x,y)= \max_{y \in Y} \min_{x \in X} f(x,y)
\]
for some \emph{convex-concave} function over $X \times Y$ where $X,
Y$ are convex compact sets. By convex-concave, we mean

\begin{definition}
A function $f$ on $X \times Y$ is \emph{convex-concave} if for every
$y \in Y$ the function $\forall x \in X, f_y(x)\triangleq f(x,y)$ is
convex on $X$ and for every $x \in X$ the function $\forall y \in Y,
f_x(y)\triangleq f(x,y)$ is concave on $Y$.
\end{definition}

Consider the case when $X$ is the set of density operators up to
some factor. Again, we will generate a series of $x^{(1)},x^{(2)},
\cdots, x^{(T)} \in X$ for $T$ rounds. For each $x^{(t)}$ in the
round $t$, we require another oracle $\mathcal{O}_2$ to find
approximate solution, denoted by $y^{(t)}$, to the following
optimization problem,
\begin{equation} \label{eqn:equilibrium_oracle}
 \max_{y \in Y} f(x^{(t)}, y)
\end{equation}
 Then $x^{(t+1)}$
will be generated via the multiplicative weight update method. After
$T$ rounds, we can claim $\frac{1}{T}\sum_{t=1}^T
f(x^{(t)},y^{(t)})$ is the approximate value for $\lambda$.
Similarly,  the oracle $\mathcal{O}_2$ is required to be
\emph{efficient solvable} and \emph{width-bounded}. By the efficient
solvability, we mean the oracle $\mathcal{O}_2$ can be solved
efficiently. By the width-boundness, we mean the
$\mathcal{L}_\infty$ norm of $y^{(t)}$ is bounded in some sense.

%
%
We will conclude this section with the proof of
Theorem~\ref{thm:mmw_generic_equilibrium}. \\ \noindent
\textbf{Proof of Theorem~\ref{thm:mmw_generic_equilibrium}}
\begin{proof}
First note that $\snorm{N^{(t)}}_\infty \leq r$ for any $1 \leq t
\leq T$. Thus,
\[
0 \leq M^{(t)}=(N^{(t)}+r\I_{\X})/2r \leq \I_\X
\]

Then it is easy to see this is a typical multiplicative weights
update method. Due to the Theorem~\ref{thm:mmw_simple}, we have:

\begin{equation}\label{eqn:mmw_main_apply_1}
 (1-\varepsilon)\sum_{\tau=1}^T \ip{\rho^{(\tau)}}{M^{(\tau)}} \leq
 \ip{\rho^{\ast}}{\sum_{\tau=1}^T M^{(\tau)}}+\frac{\ln D}{\varepsilon}
\end{equation}
for any density operator $\rho^{\ast} \in \density{\X}$. Substitute
$M^{(t)}=(N^{(t)}+r\I_{\X})/2r$ into
\ceq{\ref{eqn:mmw_main_apply_1}} and divide both side by $T$, note
that $\ip{\rho^{(t)}}{M^{(t)}}\leq 1$, then we have
\begin{equation}\label{eqn:mmw_stepone}
\frac{1}{T} \sum_{\tau=1}^{T} \ip{\rho^{(\tau)}}{N^{(\tau)}} \leq
\frac{1}{T} \ip{\rho^{\ast}}{\sum_{\tau=1}^T N^{(\tau)}} +
2r\varepsilon +\frac{2r\ln D}{\varepsilon T}
\end{equation}
By choosing $\varepsilon=\frac{\delta}{4r}$ and $T=\ceil{\frac{16r^2
\ln D}{\delta^2}}$, we have
\[
\frac{1}{T} \sum_{\tau=1}^{T} \ip{\rho^{(\tau)}}{N^{(\tau)}} \leq
\frac{1}{T} \ip{\rho^{\ast}}{\sum_{\tau=1}^T N^{(\tau)}} + \delta
\]
According to the definition of $N^{(t)}, 1\leq t \leq T$, we have
\begin{equation}\label{eqn:mmw_step_two}
\frac{1}{T} \sum_{\tau=1}^{T} \ip{S(\rho^{(t)})}{\Pi^{(t)}} \leq
\frac{1}{T} \sum_{\tau=1}^{T} \ip{S(\rho^\ast)}{\Pi^{(t)}} + \delta
\end{equation}

Choose $(\equil{\rho}, \equil{\Pi})$ to be the \emph{equilibrium
point} and substitute $\equil{\rho}$ into
\ceq{\ref{eqn:mmw_step_two}}, we have
\begin{equation}\label{eqn:mmw_stepthree}
\frac{1}{T} \sum_{\tau=1}^{T} \ip{S(\rho^{(t)})}{\Pi^{(t)}} \leq
\frac{1}{T} \sum_{\tau=1}^T \ip{S(\equil{\rho})}{\Pi^{(t)}} + \delta
\leq \equil{\lambda}+\delta
\end{equation}
where the last inequality comes from
\ceq{\ref{eqn:negative_equilibrium}}.

Since $\bar{\rho}=\frac{1}{T} \sum_{\tau=1}^T \rho^{(\tau)}$ and by
definition of $\bar{\Pi}$, we have
\[
  \ip{\bar{\rho}}{\bar{\Pi}}= \frac{1}{T} \sum_{\tau=1}^T
  \ip{\rho^{(\tau)}}{\bar{\Pi}} \leq \frac{1}{T} \sum_{\tau=1}^T
  \ip{\rho^{(\tau)}}{\Pi^{(\tau)}} \leq \equil{\lambda} + \delta
\]
where the first inequality is due to the definition of $\Pi^{(t)}$
and the second inequality comes from \ceq{\ref{eqn:mmw_step_two}}.
On the other side, $\ip{\bar{\rho}}{\bar{\Pi}} \geq \equil{\lambda}$
by the definition of $\bar{\Pi}$ and
\ceq{\ref{eqn:negative_equilibrium}}. Thus, we have
\[
 \equil{\lambda}\leq \ip{\bar{\rho}}{\bar{\Pi}} \leq
 \equil{\lambda}+ \delta
\]
Finally we need to show that this algorithm can actually run in
\class{NC} if $\delta_2=O(1/polylog(|x|))$. Consider every iteration
of the algorithm. The only operations involved are the fundamental
operation of matrices, the singular value decomposition and the
exponentials of matrices. Due to the
Fact~\ref{fact:fundamentals},\ref{fact:matrix_exponential},\ref{fact:svd},
they all can be computed in \class{NC} and the circuits for the
computation can be easily composed. Since there are only
$T=\ceil{\frac{16r^2 \ln N}{\delta^2}}$ iterations and
$r=O(polylog(|x|)), \delta=O(1/polylog(|x|))$, thus there will be at
most poly-logarithm iterations respect to the input size. Therefore,
the whole circuit will be \class{NC}.
\end{proof}

\section{QMAM case} \label{sec:qmam}

In this section, we will demonstrate how the
Framework~\ref{framework_1} can be applied to the SDP of QMAM. Due
to space limit, we will directly describe the SDP used
in~\cite{JainJUW09} and denote it by SDP (II).

\begin{center}
  \begin{minipage}{2in}
    \centerline{\underline{SDP Problem}}\vspace{-7mm}
    \begin{align*}
      \text{maximize:}\quad & \ip{R}{\rho}\\
      \text{subject to:}\quad & \tr_{\Y}(\rho) \leq \frac{1}{2}\I_\A \otimes \sigma,\\
      & \rho \in \density{\A \otimes \X \otimes \Y}, \sigma \in
      \density{\X}
    \end{align*}
  \end{minipage}
  \hspace*{12mm}
  \begin{minipage}{2in}
    \centerline{\underline{Feasibility Problem}}\vspace{-7mm}
    \begin{align*}
      \text{ask whether:}\quad & \ip{R}{\rho} \geq c\\
      \text{subject to:}\quad & \tr_{\Y}(\rho) \leq \frac{1}{2}\I_\A \otimes \sigma,\\
      & \rho \in \density{\A \otimes \X \otimes \Y}, \sigma \in
      \density{\X}
    \end{align*}
  \end{minipage}
\end{center}
where $R$ ($0\leq R\leq \I_\X$) is a POVM measurement and the space
$\A$ is of dimension 2. We need to design an algorithm to
distinguish between the following two promises. Let $\alpha$ be the
optimum value of the SDP (II).

\begin{definition} \label{def:QMAM}
Any language $L$ is inside \class{QMAM} if and only if
\begin{itemize}
  \item If $x \in L$, $\alpha\geq c(|x|)$.
  \item If $x \notin L$, $\alpha \leq s(|x|)$.
\end{itemize}
where $c(|x|)-s(|x|)=\Omega(1/poly(|x|))$.
\end{definition}

Following the Framework~\ref{framework_1}, we consider the
feasibility problem above. Precisely, we define

\begin{equation} \label{eqn:h_1}
 f_2(\{\rho,\sigma\}, \Pi)=\ip{\left(
                    \begin{array}{cc}
                       c- \ip{R}{\rho}  &  \\
                        & \tr_{\Y}(\rho)-\frac{1}{2}\I_\A \otimes \sigma \\
                    \end{array}
                  \right)}{\Pi}
\end{equation}
where $\{\rho,\sigma\} \in T_1= \density{\A \otimes \X \otimes
\Y}\times \density{\X}$ and $\Pi \in T_2=\{\Pi: 0\leq \Pi \leq
\I_{\A \otimes \X \oplus \mathbb{C}}\}$. Let $\equil{\lambda_2}$ be
the equilibrium value of function $f_2$, namely,
\[
  \equil{\lambda_2} =\min_{\{\rho,\sigma\} \in T_1} \max_{\Pi \in T_2} f_2(\{\rho,\sigma\}, \Pi) =  \max_{\Pi \in T_2} \min_{\{\rho,\sigma\} \in T_1}f_2(\{\rho,\sigma\}, \Pi)
\]
Base on the Theorem~\ref{thm:equilibrum_feasibility}, the value of
$\equil{\lambda_1}$ will imply whether the original problem is
feasible. In order to tell the two promises in
Definition~\ref{def:QMAM}, we will choose the guess value
$c=\frac{1}{2}(c(|x|)+s(|x|))$.

\begin{lemma} \label{lemma:promise_gap_qmam}
Given the two promises in Definition~\ref{def:QMAM}, we have
\begin{itemize}
  \item If $x \in L$, then $\equil{\lambda_2} \leq 0$.
  \item If $x \notin L$, then $\equil{\lambda_2} \geq \frac{1}{8}\Delta^2$.
\end{itemize}
where $\Delta=c(|x|)-s(|x|)$.
\end{lemma}

\begin{proof}
\begin{itemize}
  \item If $x \in L$, then there exists a $\rho \in \density{\A
  \otimes \X \otimes \Y}, \sigma \in \density{\X}$ such that
  $\ip{\rho}{R} \geq c$ and $\tr_{\Y}(\rho) \leq \frac{1}{2}\I_\A \otimes
  \sigma$. This implies $\equil{\lambda_2} \leq 0$.

  \item Otherwise, let $(\{\equil{\rho},\equil{\sigma}\}, \equil{\Pi})$ be the
  equilibrium point. Due to Lemma~\ref{lemma:trace_norm}, we have
  \begin{equation} \label{eqn:qmam_lemma1}
    \equil{\lambda_2}=f_2(\{\equil{\rho},\equil{\sigma}\}, \equil{\Pi})=\max \{ c-\ip{R}{\equil{\rho}},0\}
    +\frac{1}{2}\snorm{\tr_{\Y}(\equil{\rho})-\frac{1}{2}\I_\A \otimes
  \sigma}_1
  \end{equation}
  By Lemma~\ref{lemma:purification_fidelity}, there exists a $\tilde{\rho} \in \density{\A \otimes \X \otimes \Y}$
  such that $\fid(\frac{1}{2}\I_\A \otimes
  \sigma,\tr_\Y(\equil{\rho}))=\fid(\tilde{\rho}, \equil{\rho})$ and
  $\tr_\Y(\tilde{\rho})=\frac{1}{2}\I_\A \otimes
  \sigma$.
  Let $s=\frac{1}{2}\snorm{\tr_{\Y}(\equil{\rho})-\frac{1}{2}\I_\A \otimes
  \sigma}_1$ and
  $t=\frac{1}{2}\snorm{\tilde{\rho}-\equil{\rho}}_1$.
  Then if $t \leq \frac{1}{2} \Delta$, we have
  \begin{eqnarray*}
   \equil{\lambda_2} & \geq & c- \ip{R}{\tilde{\rho}}
   +\ip{R}{\tilde{\rho}-\equil{\rho}} + s \\
   & \geq & \frac{1}{2} \Delta -t +s \\
   & \geq & \frac{1}{2} \Delta -t +1 -\sqrt{1-t^2} \\
   & \geq & \frac{1}{2} \Delta -\frac{1}{2}\Delta +1
   -\sqrt{1-\frac{1}{4}\Delta^2} \geq \frac{1}{8}\Delta^2
  \end{eqnarray*}
where the first inequality is due to \ceq{\ref{eqn:qmam_lemma1}},
the second inequality comes from Lemma~\ref{lemma:trace_norm} and
the third inequality comes from Lemma~\ref{lemma:trace_bound}. The
last inequality is because $t+\sqrt{1-t^2}$ is increasing when
$0<t<\frac{1}{2}$ and $1-\sqrt{1-x^2} \geq \frac{1}{2}x^2$ for any
$0<x<1$. On the other side, if $t \geq \frac{1}{2} \Delta$, by
  \ceq{\ref{eqn:qmam_lemma1}},
\[
\equil{\lambda_2} \geq s \geq 1-\sqrt{1-t^2} \geq \frac{1}{2}t^2
\geq \frac{1}{8} \Delta^2
\]
Finally, we have $\equil{\lambda_2}\geq \frac{1}{8}\Delta^2$ in this
case.
\end{itemize}
\end{proof}

The only part left is to prove that we can calculate the equilibrium
value $\equil{\lambda_2}$ to high precision in \class{NC}. As the
readers might notice, the set $T_1$ is no longer a simple set of
density operators but a cross product of two sets of density
operators. However, we are still able to use a modified version of
the algorithm in Figure~\ref{fig:mmw_generic} to solve the problem.

\begin{figure}[t]
\noindent\hrulefill
\begin{mylist}{8mm}
\item[1.]
Let $\varepsilon=\frac{\delta}{4r}$ and $T=\ceil{\frac{16r^2 \ln
D}{\delta^2}}$. Also let $W^{(1)}=\I_{\A\otimes \X \otimes \Y}$,
$V^{(1)}=\I_{\X}$.

\item[2.]
Repeat for each $t = 1,\ldots,T$:

\begin{mylist}{8mm}
\item[(a)]
Let $\rho^{(t)}=W^{(t)}/\tr{W^{(t)}},
\sigma^{(t)}=V^{(t)}/\tr{V^{(t)}}$ and let $\Pi^{(t)}$ be the
projection onto the positive eigenspace of
$S_1(\rho^{(t)})+S_2(\sigma^{(t)})$.
\item[(b)]
Let $M_1^{(t)}=(N_1(\Pi^{(t)})+ r \I_{\A \otimes \X \otimes
\Y})/2r$, $M_2^{(t)}=(N_2(\Pi^{(t)})+r \I_\X)/2r$ and update the
weight matrix as follows:
\begin{eqnarray*}
   W^{(t+1)} & = & exp(-\varepsilon \sum_{\tau=1}^t M_1^{(\tau)}) \\
   V^{(t+1)} & = & exp(-\varepsilon \sum_{\tau=1}^t M_2^{(\tau)})
\end{eqnarray*}
\end{mylist}

\item[3.]
Return $\frac{1}{T}\sum_{\tau=1}^T
\ip{S_1(\rho^{(t)})+S_2(\sigma^{(t)})}{\Pi^{(t)}} $  as the
approximate equilibrium value of $\equil{\lambda_2}$.
\end{mylist}
\noindent\hrulefill \caption{An algorithm that computes the
approximate value of $\equil{\lambda_2}$ to precision $\delta$. }
\label{fig:mmw_qmam}
\end{figure}

Precisely, we claim the algorithm in Figure~\ref{fig:mmw_qmam} will
be able to calculate $\equil{\lambda_2}$ to precision $\delta$ in
\class{NC}. The proof is almost the same as the proof for
Theorem~\ref{thm:mmw_generic_equilibrium}. The only difference is
that we need to update $\rho^{(t)}, \sigma^{(t)}$ independently and
get two inequalities from each update. Then we combine them to get
the final result. This can be done because
\[
 f_2(\{\rho,\sigma\}, \Pi)=\ip{S_1(\rho)}{\Pi}+ \ip{S_2(\sigma)}{\Pi}
\]
where
\[
S_1(\rho)=\left(
                    \begin{array}{cc}
                       c- \ip{R}{\rho}  &  0\\
                     0   & \tr_{\Y}(\rho) \\
                    \end{array}
                  \right)  \text{ and } S_2(\sigma)=\left(
                                       \begin{array}{cc}
                                       0    & 0 \\
                                        0   & -\frac{1}{2}\I_\A \otimes \sigma \\
                                       \end{array}
                                     \right)
\]

Let $\Pi=\left(
           \begin{array}{cc}
             p &  \\
               & P \\
           \end{array}
         \right)$ again, we can choose
\[
  N_1(\Pi)= -pR+P\otimes \I_\Y+pc \I_{\A \otimes \X \otimes \Y}
  \text{  and } N_2(\Pi)=-\frac{1}{2} \tr_\A P
\]
It is easy to verify that $\snorm{N_1(\Pi)}_\infty,
\snorm{N_2(\Pi)}_\infty$ is bounded by 3. Thus given the gap of the
equilibrium value $\equil{\lambda_2}$ between two promises in
Lemma~\ref{lemma:promise_gap_qmam}, we can distinguish them in
\class{NC(poly)} namely \class{PSPACE}.

\section{Precision Issue} \label{sec:precision_issue}

The discussions on precision issue about the parallel
implementations in our paper are quite similar to the arguments used
in~\cite{JainJUW09,JainUW09,JainW09,Wu10}. Again, we will make use
the three facts in Appendix~\ref{sec:NC_parellel} and truncate the
computation to sufficient precision for each step. It should be
noticed that because of our new framework we will only use three
types of operations of matrices (namely, the three facts in
Appendix~\ref{sec:NC_parellel}). The analysis of the precision issue
in our paper then should be easier than the one in~\cite{JainJUW09}
since the latter one involves other types of operations, like the
inversion of matrices.

\end{document}